\documentclass[12pt]{article}
\usepackage{mymacros}

\title{Complex crystallographic reflection 
groups and Seiberg--Witten integrable systems: rank 1 case}

\makeatletter
\renewcommand\@date{{
  \vspace{-\baselineskip}
  \large\centering
  \begin{tabular}{@{}c@{}}
    Philip C. Argyres\textsuperscript{$a$} 
  \end{tabular}
  \!,
  \begin{tabular}{@{}c@{}}
    Oleg Chalykh\textsuperscript{$b$} 
  \end{tabular}
  \!,
  \begin{tabular}{@{}c@{}}
    Yongchao L\"u\textsuperscript{$c$} 
  \end{tabular}
  \bigskip

  {\small \textsuperscript{$a$}\,Physics Department, 
	 University of Cincinnati,
	 PO Box 210011, 
	 Cincinnati OH 45221, US} \\
    \normalsize {\small philip.argyres@gmail.com} \par
  {\small \textsuperscript{$b$}\,School of Mathematics, 
	University of Leeds, 
	Leeds, LS2 9JT, UK } \\
    \normalsize {o.chalykh@leeds.ac.uk} \par
 {\small
	 \textsuperscript{$c$}\,School of Physics, Korea Institute for Advanced Study, Seoul 02455, Korea} \\
    \normalsize {yongchaolu@kias.re.kr}
}}
\makeatother

\begin{document}
\maketitle

\begin{abstract}
We consider generalisations of the elliptic Calogero--Moser systems associated to complex crystallographic groups in accordance to \cite{EFMV11ecm}. In our previous work \cite{Argyres:2021iws}, we proposed these systems as candidates for Seiberg--Witten integrable systems of certain SCFTs. Here we examine that proposal for complex crystallographic groups of rank one. Geometrically, this means considering elliptic curves $T^2$ with $\Z_m$-symmetries, $m=2,3,4,6$, and Poisson deformations of the orbifolds $(T^2\times\mathbb{C})/\Z_m$. The $m=2$ case was studied in \cite{Argyres:2021iws}, while $m=3,4,6$ correspond to Seiberg--Witten integrable systems for the rank 1 Minahan--Nemeschansky SCFTs of type $E_{6,7,8}$. This allows us to describe the corresponding elliptic fibrations and the Seiberg--Witten differential in a compact elegant form. This approach also produces quantum spectral curves for these SCFTs, which are given by Fuchsian ODEs with special properties.     
\end{abstract}

\tableofcontents
\section{Introduction}
The present paper is a study of Cherednik algebras and integrable systems associated with rank-one complex crystallographic groups. One of the motivations for this work is to uncover a connection of these integrable systems to quantum field theories. To provide some context, we recall that the study of supersymmetry with eight supercharges has proven to be a fruitful area of research, providing valuable insights into the strong coupling dynamics of quantum field theory. Among the earliest examples of such theories are the Minahan--Nemeschansky theories \cite{Minahan:1996fg, Minahan:1996cj}, which have inspired decades of flourishing development. Despite their lack of conventional Lagrangian descriptions, many interesting observables have been calculated, and their further study is expected to provide new insights. One of the most promising avenues for such study is through the Seiberg--Witten solutions \cite{SW941, SW942}, which exhibit a close relationship to integrable systems \cite{DW95}. 
Seiberg--Witten integrability, as it has become known, has provided profound insight into the strong coupling dynamics of field theories. However, there is no systematic method to recognise a Seiberg--Witten integrable system behind a given quantum field theory. Moreover, sometimes there is more than one possibility, typically due to the proposed integrable models being geometrically equivalent despite their rather different origins. 

In our previous paper \cite{Argyres:2021iws}, we proposed to study the so-called {\it crystallographic elliptic Calogero--Moser systems} as a potential source of Seiberg--Witten integrable systems. Recall that, according to \cite{EFMV11ecm}, each complex crystallographic reflection group has an associated family of elliptic Calogero--Moser systems constructed using the theory of {\it elliptic Cherednik algebras}. We expect that many of these systems can be viewed as Seiberg--Witten integrable systems for certain superconformal field theories (SCFTs). In \cite{Argyres:2021iws}, this was partly confirmed  for the Inozemtsev system, a version of the elliptic
Calogero--Moser system associated with the non-reduced root system of type $BC_n$, or with the complex crystallographic group $[G(1,2,n)]^{\alpha}_1$ in Popov's classification \cite[Theorem 2.6.1]{Popov82}.


In the present paper we turn our attention to groups (and SCFTs) of {\it rank one}. Geometrically, this means considering elliptic curves $\E$ with $\Z_m$-symmetries, $m=2,3,4,6$, and Poisson deformations of the orbifolds $T^*\E/\Z_m$. Such deformations can be constructed using the classical version of the Cherednik algebra, $H_c$, and its spherical subalgebra, $B_c$. Here $c$ denotes the set of deformation parameters whose number equals $4$, $6$, $7$ and $8$ for $m=2,3,4$ and $6$, respectively. The integrable system appears as the (Poisson-commutative) algebra of the global sections of $B_c$. In our case, it is described by a single hamiltonian function $h$ on $T^*\E$, which can also be viewed as a function on $T^*\E/\Z_m$ due to its symmetry. 

Our main result describes the fibration on $T^*\E/\Z_m$ by the level sets of the hamiltonian $h$ as an {\it elliptic pencil} of a special form. As a corollary, the varieties $\mathrm{Spec}\,B_c$ ({\it Calogero--Moser spaces}) can be identified in this case as certain {\it rational elliptic surfaces}. In fact, these surfaces also admit an interpretation as certain Dolbeault moduli spaces, $\mathcal M_{Dol}$, of (weakly) parabolic Higgs bundles over a punctured sphere. The corresponding de Rahm spaces, $\mathcal M_{dR}$, were considered by P.~Boalch in \cite{boalch09} in relation to affine Dynkin quivers, Kronheimer's ALE spaces, and difference Painleve equations. However, their connection to Cherednik algebras and elliptic integrable systems is new. Additionally, we study a natural quantisation of the constructed elliptic fibrations: they are given by pencils of Fuchsian ODEs, for which we give a full characterization. This provides vivid examples of (not yet fully established in the parabolic case) mirror symmetry for Hitchin fibrations and quantisation of their spectral curves (see Section \ref{sec6} below).    

Turning to the supersymmetric quantum field theories, one would like to associate the geometric picture above with some rank 1 4d $\mathcal{N}=2$ SCFTs. For the rich landscape and classification of such theories, see \cite{ArgyresRankOne1, ArgyresRankOne2, ArgyresRankOne3, ArgyresRankOne4}. As it turns out, our $m=2$ case corresponds to the well-known ${\rm SU}(2)$ superconformal gauge theory with 4 flavors which possesses a $D_4$ flavor symmetry \cite{DW95}. Meanwhile, the $m=3, 4, 6$ cases correspond to Minahan-Nemeschansky (MN) theories \cite{Minahan:1996fg, Minahan:1996cj} that possess the exceptional flavor symmetry of types $E_6$, $E_7$, $E_8$ (note that these theories lack conventional Lagrangian descriptions).   
This new perspective allows us to obtain the elliptic fibration and the Seiberg--Witten differential for the rank 1 MN theories in a systematic fashion. Notably, the mass parameters of those rank 1 SCFTs receive a transparent geometric interpretation, at the same time being directly linked to the deformation parameters of the elliptic Cherednik algebras. As a byproduct, we find a natural quantisation of the Seiberg--Witten curves of the rank 1 Minahan--Nemeschansky theories. It is worth mentioning that the elliptic fibrations in Weierstrass form have already been established for these theories \cite{SW942, Minahan:1996fg, Minahan:1996cj}, but they seem less suitable for quantisation. 

Last but not least, our results pave the way to constructing Seiberg--Witten curves of higher-rank Minahan--Nemeschansky theories, which will be done in the subsequent work \cite{ACL2}.   

The organisation of the paper is as follows. In Section \ref{sec2} we consider Cherednik algebras for elliptic curves with symmetries, and describe the hamiltonians of the corresponding integrable systems. Section \ref{sec3} describes the classical dynamics of these integrable systems in geometric terms, using their Lax form. We observe a peculiar {\em duality} of the Lax matrix, which leads to a compact formula for its spectral curve. This produces elliptic fibrations and Seiberg--Witten differentials for the appropriate SCFTs. In Section \ref{pencils}, we interpret these fibrations in terms of elliptic pencils of a particular form. In Section \ref{sec5} we introduce quantum spectral curves by passing from the classical to the quantum hamiltonian. We further characterise the resulting families of Fuchsian ODEs. Finally, in Section \ref{sec6} we discuss other contexts in which the related structures appear. The paper finishes with four appendices giving further details, explicit formulas, and additional properties of the classical and quantum spectral curves obtained in Sections \ref{pencils} and \ref{sec5}.

\section{Complex crystallographic groups and Cherednik algebras in rank one}\label{sec2}

If $X$ is a complex manifold with an action of a finite group $W$, that action naturally extends to the sheaf $\mathcal D[X]$ of regular differential operators on $X$. Therefore, one may consider the crossed
product $\mathcal{D}[X]\rtimes W$ and the subsheaf of $W$-invariants
$\mathcal{D}[X]^W$, both regarded as sheaves of algebras over $X/W$. In such situation,  Etingof constructs in \cite{Etingof04} the global Cherednik algebra $H_{c}(X, W)$ and its spherical subalgebra $B_{c}(X, W)$ as certain (in fact, universal) deformations of $\mathcal D[X]\rtimes W$ and $\mathcal D[X]^W$, respectively. A special case of interest is when $X={\mathbb C}^n/\L$ is a complex torus, and $W\subset\mathrm{GL}_n({\mathbb C})$ is a complex reflection group preserving the lattice $\L$, thus acting on $X$. The semi-direct product $G=\L\rtimes W$ is an example of a {\it complex crystallographic reflection group}, namely a discrete
subgroup of the complex affine group of $\mathbb{C}^n$ 
generated by (affine) reflections, 
whose linear part
$W$ is a finite complex reflection group and whose translation lattice $\Gamma$ is of full rank $2n$ in $\mathbb{C}^n$;
all such groups are classified in \cite{Popov82}\footnote{In general, $G$ may not be a semidirect product of $W$ and $\Gamma$. Also, one may consider a more general situation than in \cite{Popov82} so that $W$ is generated by reflections but $G$ may not be (that is how, for example, extended affine Weyl groups arise).}. In this case $H_{c}(X, W)$ is referred to as the {\it elliptic Cherednik algebra}. The significance of $X$ being a complex torus is that in that case, according to \cite{EFMV11ecm}, the spherical subalgebra $B_{c}(X, W)$ has a commutative subalgebra of dimension $n$. This defines a family of integrable systems on $X$, called {\it crystallographic elliptic Calogero--Moser systems}. In this paper we look into the simplest cases, namely, complex crystallographic  groups and elliptic Cherednik algebras of rank 1, corresponding to elliptic curves with symmetries.  

\subsection{Elliptic curves with symmetries}

Let $\E={\mathbb C}/\L$ with $\L=2\omega_1\Z+2\omega_2\Z$ be an elliptic curve. We use a $q\in{\mathbb C}$ to represent a point on both ${\mathbb C}$ and $\E$. We follow the standard convention assuming $\mathrm{Im}\,({\omega_2}/{\omega_1})>0$. In general, the only holomorphic automorphisms (symmetries) of $\E$ are: (1) translations, or (2) translations followed by the $\Z_2$-symmetry $q\mapsto -q$. Elliptic curves with larger automorphism groups arise when $\L$ has a rotational symmetry of order $m>2$. As is well known, the only possibilities are $m=3,4,6$, and the groups $G=\L\rtimes \Z_m$ with $m=2,3,4,6$ (plus the trivial case $G=\L$) exhaust all complex crystallographic groups of rank one. Let us choose $\omega_{1,2}$ so that 
\begin{equation}
\omega_2/\omega_1=e^{\pi i/3} \quad\text{(when $m=3,6$)}\quad\text{or}\quad \omega_2/\omega_1=e^{\pi i/2}\quad\text{(when $m=4$)}\,.    
\end{equation}
The first case is known as \emph{equianharmonic} (with the hexagonal lattice $\L$); the $m=4$ case is called \emph{lemniscatic} (with the square lattice $\L$). In each case, we think of $\E$ as having an extra symmetry
\begin{equation}\label{sm}
 s: q\mapsto \omega q\,,\qquad \omega=e^{2\pi i/m}\,,  
\end{equation}
and write $\Z_m:=\{1,s,\dots s^{m-1}\}$ for the multiplicative group of the $m$-th roots of unity, acting on $\E$. The generic $\E$ corresponds to $m=2$. The point $q=0$ is always fixed by $\Z_m$; other fixed points and their stabiliser groups are given in the table \ref{table0}.  These are also shown in Fig.~\ref{fig:orbifoldtori}.
\begin{table}[H]
\centering
\begin{align*}
\renewcommand\arraystretch{1.2}
\begin{array}{c|c|c }
\hline
   \qquad m \qquad   & \qquad \text{fixed points $\ne 0$} \qquad &  \qquad \text{stabilisers}  \qquad \\
    \hline
  2 & \omega_{1,2,3} & \Z_2 \\
\hline
    3 & \eta_{1,2} & \Z_3 \\
\hline
 4 & \begin{array}{c}
          \omega_{1,2} \\
     \omega_3
 \end{array}
    & \begin{array}{c}
        \Z_2 \\
        \Z_4
    \end{array} \\
\hline
    6 & \begin{array}{c}
          \omega_{1,2,3} \\
     \eta_{1,2}
 \end{array} & \begin{array}{c}
          \Z_2 \\
     \Z_3
 \end{array} \\
      \hline
\end{array}
\end{align*}
 \caption{Non-zero fixed points and their stabiliser groups. Here we use the notation $\omega_3=\omega_1+\omega_2$, $\eta_{1}=2\omega_3/3$, $\eta_2=2\eta_1$.}
\label{table0}
\end{table}

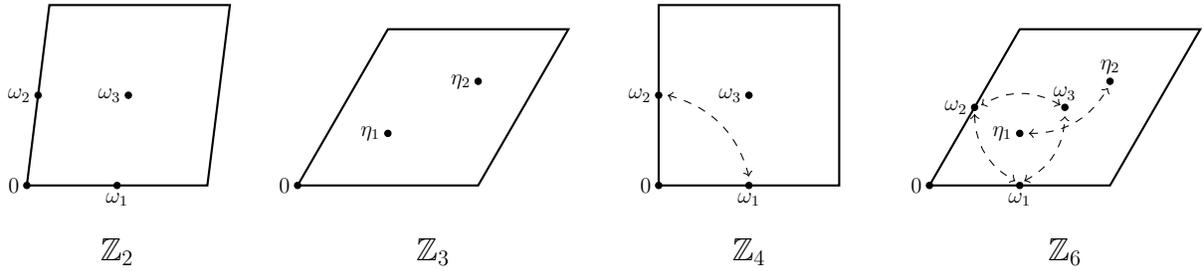
\begin{figure}[h]
\centering
\begin{tikzpicture}[scale=0.6, every node/.style={scale=0.7}]

\begin{scope}[xshift=-7 cm]
\draw[thick] (0,0)--(4,0)--(4.5,4)--(0.5,4)--(0,0);
\filldraw[black] (0,0) circle (2pt) node[anchor=east]{$0$};
\filldraw[black] (2.25, 2) circle (2pt) node[anchor=east]{$\omega_3$};
\filldraw[black] (2,0) circle (2pt) node[anchor=north]{$\omega_1$};
\filldraw[black] (0.25,2) circle (2pt) node[anchor=east]{$\omega_2$};
\node[anchor=north] at (2,-1){ \Large $\mathbb Z_2$};
\end{scope}

\begin{scope}[xshift=-1 cm]
\draw[thick] (0,0)--(4,0)--(6, {2*sqrt(3)})--(2, {2*sqrt(3)})--(0,0);
\filldraw[black] (0,0) circle (2pt) node[anchor=east]{$0$};
\filldraw[black] (2, {2*sqrt(3)/3}) circle (2pt) node[anchor=east]{$\eta_1$};
\filldraw[black] (4, {4*sqrt(3)/3}) circle (2pt) node[anchor=east]{$\eta_2$};
\node[anchor=north] at (3,-1){ \Large $\mathbb Z_3$};
\end{scope}

\begin{scope}[xshift=7 cm]
\draw[thick] (0,0)--(4,0)--(4,4)--(0,4)--(0,0);
\filldraw[black] (0,0) circle (2pt) node[anchor=east]{$0$};
\filldraw[black] (2, 2) circle (2pt) node[anchor=east]{$\omega_3$};
\filldraw[black] (2,0) circle (2pt) node[anchor=north]{$\omega_1$};
\filldraw[black] (0,2) circle (2pt) node[anchor=east]{$\omega_2$};
\draw[<->,dashed] (0.2,2)  to [bend left] (2, 0.2);
\node[anchor=north] at (2,-1){ \Large $\mathbb Z_4$};
\end{scope}

\begin{scope}[xshift=13 cm]
\draw[thick] (0,0)--(4,0)--(6, {2*sqrt(3)})--(2, {2*sqrt(3)})--(0,0);
\filldraw[black] (0,0) circle (2pt) node[anchor=east]{$0$};
\filldraw[black] (2, {2*sqrt(3)/3}) circle (2pt) node[anchor=east]{$\eta_1$};
\filldraw[black] (4, {4*sqrt(3)/3}) circle (2pt) node[anchor=south]{$\eta_2$};
\filldraw[black] (2,0) circle (2pt) node[anchor=north]{$\omega_1$};
\filldraw[black] (3, {sqrt(3)}) circle (2pt) node[anchor=south]{$\omega_3$};
\filldraw[black] (1, {sqrt(3)}) circle (2pt) node[anchor=east]{$\omega_2$};
\draw[<->,dashed] (2.2, {2*sqrt(3)/3})  to [bend right] (4-0.05, {4*sqrt(3)/3 -0.15});
\draw[<->,dashed] (2.1,0.1)  to [bend right] (3, {sqrt(3)-0.2});
\draw[<->,dashed] (3-0.1, {sqrt(3)+0.05})  to [bend right] (1.15, {sqrt(3)+0.05}) ;
\draw[<->,dashed] (1, {sqrt(3)-0.15})   to [bend right] (2-0.1,0.1);
\node[anchor=north] at (3,-1){ \Large $\mathbb Z_6$};
\end{scope}

\end{tikzpicture}
\caption{Fundamental domain, group action, and fixed points.}
\label{fig:orbifoldtori}
\end{figure}

Let $\sigma(q)=\sigma(q|2\omega_1, 2\omega_2)$, $\zeta(q)=\zeta(q|2\omega_1, 2\omega_2)$, $\wp(q)=\wp(q|2\omega_1, 2\omega_2)$ be the Weierstrass $\sigma$, $\zeta$ and $\wp$ functions associated with $\L$ and $\E$. Since $\omega \L=\L$, we have
\begin{equation}
    \sigma(\omega q)=\omega\sigma(q)\,,\quad \zeta(\omega q)=\omega^{-1}\zeta(q)\,,\quad \wp(\omega q)=\omega^{-2}\wp(q)\,. 
\end{equation}
The general Weierstrass form of $\E$, $\wp'^2=4\wp^3-g_2\wp-g_3$ (for $m=2$), specialises to  
\begin{align}
\wp'^2&=4\wp^3-g_3\qquad(m=3,\,6)\,,
\\
\wp'^2&=4\wp^3-g_2\wp\qquad(m=4)\,.
\end{align}

The quotient $\E/\Z_m$ is isomorphic to $\P^1$, which allows us to view $\E$ as an $m$-fold branched covering of $\P^1$. Namely, 
\begin{equation}\label{mcov}
  v^m=P_m(u)\,,  
\end{equation}
where the elliptic functions $u,v$ and the corresponding polynomial $P_m$ are summarised in table \ref{table1}.

\begin{table}[h]
\centering
\begin{align*}
\renewcommand\arraystretch{1.6}
\begin{array}{c | c c c c}
\hline
     m   & 2 &  3 & 4 & 6 \\
    \hline
    \omega_2/\omega_1 & \text{any} & e^{{ \pi i }/{3}} & i &  e^{{ \pi i }/{3}} \\
    \hline
    u & \wp(q) & \frac12\wp'(q) & \wp^2(q) & \wp^3(q) \\
    \hline
    v & \frac12\wp'(q) &  \wp(q) & \frac12\wp'(q) & \frac12\wp(q) \wp'(q) \\
    \hline
    P_m(u) & u^3 - \frac14g_2 u - \frac14g_3 & u^2 + \frac14g_3 & u( u-\frac14g_{2})^{2} & u^{2}(u-\frac14g_{3})^{3} \\
    \hline
\end{array}
\end{align*}
 \caption{The elliptic functions $u,v$ and the corresponding polynomial $P_m$.}
\label{table1}
\end{table}

Thus, with appropriate $e_1,e_2,e_3$,
\begin{align}\label{cci2}
v^2&=(u-e_1)(u-e_2)(u-e_3) &(m=2)&\,,\\\label{cci3}
v^3&=(u-e_1)(u-e_2) &(m=3)&\,,\\\label{cci4}
v^4&=(u-e_1)(u-e_2)^2 &(m=4)&\,,\\\label{cci6}
v^6&=(u-e_1)^2(u-e_2)^3 &(m=6)&\,.
\end{align}
These curves have three or four branch points (one of these being $e_0=\infty$), and they can be recognised as the only {\it genus-one cyclic coverings} of $\P^1$.

The action \eqref{sm} naturally extends to a symplectic $\Z_m$-action on $T^*\E$, so we may consider the orbifold $T^*\E/\Z_m$. A deformation of this orbifold can be constructed using {\it elliptic Cherednik algebras}; this deformation plays the central role in our work.

\subsection{Rational Cherednik algebra for $W=\Z_m$}
\label{rca1} 
Write $\mathcal D$ for the ring of differential operators in $q\in{\mathbb C}$, with meromorphic coefficients. The group $\Z_m=\{1, s, \dots, s^{m-1}\}$ acts by 
\begin{equation}
s:\,q\mapsto\omega q\,,\quad \omega=e^{2\pi i/m}\,,
\end{equation} 
and this action naturally extends to $\mathcal D$. We then have the following relations in the crossed product $\mathcal D\rtimes \Z_m$: 
\begin{equation}
    sq=\omega qs\,,\quad s\frac{d}{dq}=\omega^{-1}\frac{d}{dq}s\,,\quad \frac{d}{dq}q=q\frac{d}{dq}+1\,.
\end{equation}
We may think of the elements in $\mathcal D\rtimes \Z_m$ as acting on functions of $q$.    

For $c:=(c_1,\dots, c_{m-1})\in{\mathbb C}^{m-1}$ and $\hbar\ne 0$, define the Cherednik algebra 
$H_{\hbar, c}=H_{\hbar,c}(\Z_m)$ as the subalgebra in $\mathcal D\rtimes \Z_m$ generated by $\C[q]\rtimes \Z_m$ and the {\it Dunkl operator},
\begin{equation}\label{du}
    y=\hbar\frac{d}{dq}-\sum_{l=1}^{m-1}\frac{c_l}{q}\,s^l\,.
\end{equation}
One can view such $y$ as a deformation of $\hbar\frac{d}{dq}$, with the deformation paramters  $c_l\in\mathbb{C}$ associated with the elements in $\mathbb{Z}_m\setminus\{1\}$.

Alternatively, $H_{\hbar, c}$ can be described as the associative algebra generated by $q,y,s$ subject to
\begin{equation}\label{rel}
    s^m=1\,,\quad sq=\omega qs\,,\quad sy=\omega^{-1}ys\,,
    \quad yq-qy=\hbar +\sum_{l=1}^{m-1}(1-\omega^l)c_ls^l\,. 
\end{equation}
  
The {\it spherical subalgebra} of $H_{\hbar, c}$ is defined as $e\,H_{\hbar, c}\,e$, where
\begin{equation}
  e=\frac1m\sum_{j=0}^{m-1}{s^j}\,.  
\end{equation}
Each element of the spherical subalgebra, when acting on $\Z_m$-invariant polynomials ${\mathbb C}[q^m]$, reduces to a $\Z_m$-invariant differential operator; this defines a faithful representation 
\begin{equation}\label{res}
\theta\,:\  e\,H_{\hbar, c}\,e\ \longrightarrow \ \mathcal D^{\Z_m}\,,  
\end{equation} 
called the {\it Dunkl representation}.
Introduce
\begin{equation}\label{mui}
    \mu_i=\sum_{l=1}^{m-1}\omega^{-il}c_l\,,\quad i=0,\dots, m-1\,,
\end{equation}
\begin{equation}\label{abc}
u=q^m\,,\quad
    v=\hbar q\frac{d}{dq}\,,\quad w=\left(\hbar \frac{d}{dq}-\frac{\mu_{m-1}}{q}\right)\dots \left(\hbar \frac{d}{dq}-\frac{\mu_{0}}{q}\right) \,.
\end{equation}
Then one finds that under $\theta$, the elements $e\,q^m\,e$, $e\,y^m\,e$, and $e\,qy\,e$ are mapped to $u, w$, and $v-\mu_0$, respectively.
We denote by $B_{\hbar, c}$ the spherical subalgebra in the Dunkl representation:
\begin{equation}
B_{\hbar, c}=\theta(e\,H_{\hbar, c}\,e)\,,\qquad B_{\hbar, c}\subset \mathcal D^{\Z_m}\,.
\end{equation}
It is easy to show that, as an abstract algebra, $B_{\hbar, c}$ is generated by $u, v, w$ subject to the relations 
\begin{equation}\label{relsph}
    [v,u]=\hbar m u\,,\quad [w,v]=\hbar m w\,,\quad uw=P_\hbar(v)\,,\quad wu=P_\hbar(v+\hbar m)\,,
\end{equation}
where $P_\hbar(t)=\prod_{j=0}^{m-1}\left(t-\mu_j-j\hbar\right)$. 

Setting $\hbar=0$ in \eqref{rel} and \eqref{relsph}, we obtain the {\it classical analogues}, $H_{0,c}$ and $B_{0,c}$. 
The algebra $B_{0,c}$ can therefore be described abstractly as the quotient
\begin{equation}
B_{0,c}={{\mathbb C}[\bar u,\bar v,\bar w]}/{\{\bar u\bar w=P(\bar v)\}}\,,\qquad P(t)=\prod_{i=0}^{m-1}(t-\mu_i)\,.     
\end{equation}
When all $\mu_i=0$, this is the algebra of functions on the {\it cyclic singularity} ${\mathbb C}^2/\Z_m$. Therefore, the family $B_{0,c}$ describes a Poisson deformation of the cyclic singularity, with the Poisson bracket induced from \eqref{relsph},  
\begin{equation}\label{relsphc}
    \{\bar v, \bar u\}=\bar u\,,\quad \{\bar w,\bar v\}= \bar w\,,\quad \{\bar w, \bar u\}=P'(\bar v)\,.
\end{equation}
We remark that the algebras $H_{0,c}$ and $B_{0,c}$ can also be constructed similarly to $H_{\hbar,c}$ and $B_{\hbar,c}$, replacing the ring $\mathcal D$ by the commutative ring ${\mathbb C}(q)[p]$, where $p$ replaces $\pp:=\hbar\frac{d}{dq}$. The classical analogues of the Dunkl operator \eqref{du} and of the generators \eqref{abc} are:
\begin{equation}\label{duc}
    y^c=p-\sum_{l=1}^{m-1}\frac{c_l}{q}\,s^l\,,\qquad
 \bar u=q^m \,,\quad
    \bar v=qp\,,\quad  \bar w=\left(p-\frac{\mu_{m-1}}{q}\right)\dots \left(p-\frac{\mu_{0}}{q}\right)\,.  
\end{equation}

\subsection{Elliptic Cherednik algebra of rank one}\label{Ellversion}

We now proceed to define the elliptic version of $H_{\hbar, c}$, following the general framework of \cite{Etingof04}. Let $\E=\C/\L$ be an elliptic curve with the symmetry group $\Z_m$, $m\in\{2,3,4,6\}$. The elliptic Cherednik algebra $H_{\hbar, c}(\E)$ depends on a set of parameters chosen as follows. To every $x_i\in\E$ and $l=1,\dots ,m-1$ such that $s^l(x_i)=x_i$, we assign a parameter $c_l(x_i)$, with an additional requirement that $c_l(x_i)=c_l(x_j)$ whenever $x_j=wx_i$ for some $w\in W$. From Fig.~\ref{fig:orbifoldtori} we observe that this amounts to $4$ parameters in the $m=2$ case, and $6$, $7$ or $8$ parameters when $m=3$, $4$ or $6$, respectively. We write $c=(c_l(x_i))$ for the set of parameters. It will be convenient to extend the set $c$ by setting $c_l(x_i)=0$ if $s^l(x_i)\neq x_i$. Later it will also be convenient to use the following combinations:
\begin{equation}\label{emu}
    \mu_j(x_i)=\sum_{l=1}^{m-1}\omega^{-jl}c_l(x_i)\,,\quad j=0,\dots, m-1\,.
\end{equation}
Note that for any fixed point $x_i$, the sum of $\mu_j(x_i)$ is zero, and if the stabiliser of $x_i$ is a proper subgroup $\Z_{m_i}\subset \Z_m$ then the set of $\mu_j(x_i)$ has repetitions.

Let $\C(\E)$ denote the field of meromorphic functions on $\E$ (i.e. elliptic functions in $q$), and $\mathcal D_\E$ the ring of differential operators on $\E$, with elliptic coefficients. 
Similarly to the rational case, we form the cross-product $\mathcal D_\E\rtimes\Z_m$.
To define $H_{\hbar, c}(\E)$ as a sheaf of algebras over $\E/\Z_m$, we need to describe its sections over an arbitrary $\Z_m$-invariant open chart $U\subset \E$. Write $\mathcal O_U\subset \C(\E)$ for the ring of functions regular on $U$.
We define the {\it algebra of sections} of $H_{\hbar, c}(\E)$ over $U$ to be  
the subalgebra $H_{\hbar, c}(\E, U)\subset \mathcal D_\E\rtimes \Z_m$ generated by $\mathcal O_U\rtimes \Z_m$ and 
an element $y$ of the form
\begin{equation}\label{gdu}
    y=\hbar\frac{d}{d q}-\sum_{l=1}^{m-1}b_l(q)s^l\,,
\end{equation}
where $b_l(q)$ are allowed simple poles at the fixed points $x_i\in U$, with
\begin{equation}\label{resdu}
 \mathrm{res}_{q=x_i}b_l=c_l(x_i)\quad\text{for all $x_i\in U$.}
\end{equation}
Informally, these conditions mean that near $q=x_i$ such $y$ should look like the rational Dunkl operator \eqref{du}. Note that while there may be several such elements $y$, the difference of any two of them belongs to $\mathcal O_U\rtimes\Z_m$, so our definition of $H_{\hbar, c}(\E, U)$ is unambiguous.  

This defines the sheaf $H_{\hbar, c}(\E)$ of {\it elliptic Cherednik algebras} on $\E$. The sheaf of {\it spherical subalgebras} $eH_{\hbar, c}(\E)e$ is obtained by replacing local sections $a\in \mathcal D_\E\rtimes\Z_m$ by $eae$. Again, these local sections can be realised as differential operators using the map \eqref{res}. This produces a sheaf of algebras $B_{\hbar, c}(\E):=\theta \left(eH_{\hbar, c}(\E)e\right) \subset\mathcal D_\E^{\Z_m}$.

\bigskip

The {\it classical version} $H_{0,c}(\E)$ of the elliptic Cherednik algebra is obtained in a similar fashion. Namely, the classical counterpart of the ring of differential operators is the commutative ring $\C(\E)[p]:=\C(\E)\otimes\C[p]$, with the $\Z_m$-action extended by $sp=\omega^{-1}ps$. The algebra of sections $H_{0,c}(U)$ over any open $\Z_m$-invariant chart $U\subset \E$ is defined as the subalgebra of $\C(\E)[p]\rtimes \Z_m$, generated by $\mathcal O_U\rtimes \Z_m$ and an element $y^c$ of the form
\begin{equation*}
    y^c=p-\sum_{l=1}^{m-1}b_l(q)s^l\,,
\end{equation*}
where $b_l(q)$ satisfy the same conditions \eqref{resdu}. The sheaves of classical spherical subalgebras $eH_{0, c}(\E)e$ and $B_{0, c}(\E)$ are defined in a similar way. 
\medskip
\begin{remark}
When $c=0$, $H_{\hbar, 0}=\mathcal D[\E]\rtimes \Z_m$ and $B_{\hbar, 0}=\mathcal D[\E]^{\Z_m}$, where $\mathcal{D}[\E]$ denotes the sheaf of regular differential operators on $\E$. The classical analogue of $\mathcal D[\E]$ is the sheaf $\mathcal O(T^*\E)$, so we get $H_{0,0}=\mathcal O(T^*\E)\rtimes \Z_m$ and $B_{0,0}=\mathcal O(T^*\E)^{\Z_m}$. The latter sheaf can be identified with $\mathcal O(T^*\E/\Z_m)$, hence, the sheaf $B_{0,c}$ describes a deformation of the orbifold $T^*\E/\Z_m$. As we will see below, these deformations can be described geometrically as certain {\it rational elliptic surfaces}.    
\end{remark}

\subsection{Case of a punctured elliptic curve}

For $U=\E_0:=\E\setminus\{0\}$ the algebra $H_{\hbar, c}(\E_0)$ admits a simple description. Denote by $f(x,z)$ the following elliptic function:
\begin{equation}
    f(x,z)=\zeta(x-z)-\zeta(x)+\zeta(z)=\frac{1}{2}\frac{\wp'(x)+\wp'(z)}{\wp(x)-\wp(z)}\,.
\end{equation}
Note that $f(\omega x, \omega z)=\omega^{-1}f(x,z)$ for $\omega=e^{2\pi i/m}$.

Write $\mathcal S\subset \{\omega_{1,2,3}, \eta_{1,2}\}$ for the set of nonzero fixed points for $\Z_m$, and consider
\begin{equation}
    y=\pp-\sum_{l=1}^{m-1}\sum_{x_i\in\mathcal S}c_l(x_i)f(q, x_i)s^l\,,\qquad \pp=\hbar\frac{d}{d q}\,.
\end{equation}
It is clear that $y$ belongs to $H_{\hbar, c}(\E_0)$. 
By the $\Z_m$-invariance of the couplings $c_l(x_i)$, 
\begin{equation}\label{rele1}
    sy=\omega^{-1}ys\,.
\end{equation}
Next, we have $\mathcal O_{\E_0}=\C[\wp, \wp']/\{\wp'^2=4\wp^3-g_2\wp-g_3\}$ and the crossed product $\mathcal O_{\E_0}\rtimes \Z_m$, with the relations 
\begin{equation}
    s\wp=\omega^{-2}\wp s\,,\quad s\wp'=\omega^{-3}\wp's\,.
\end{equation}
By definition, $H_{\hbar, c}(\E_0)$ is generated by $\mathcal O_{\E_0}\rtimes \Z_m$ and $y$. For any $g\in\mathcal O_{\E_0}$,
\begin{equation}\label{rele2}
    yg- gy= \hbar \frac{dg}{dq} - \sum_{l=1}^{m-1}\sum_{x_i\in\mathcal S}c_l(x_i)f(q, x_i)(g(\omega^{l}q)-g(q)) s^l\,.
\end{equation}
Note that whenever $x_i\in\mathcal S$ is fixed by $s^l$, the function $g(\omega^{l}q)-g(q)$ vanishes at $q=x_i$, and if $x_i$ is not fixed by $s^l$, then $c_l(x_i)=0$.
Therefore, all the terms in the right-hand side of \eqref{rele2} are regular away from zero, hence, belong to $\mathcal O_{\E_0}\rtimes \Z_m$. It is easy to show that as an abstract algebra, $H_{\hbar, c}(\E_0)$ is generated by $\mathcal O_{\E_0}\rtimes \Z_m$ and $y$, subject to the relations \eqref{rele1}, \eqref{rele2}. 

The algebra $H_{\hbar, c}(\E_0)$ admits a basis formed by the elements
\begin{equation}
    y^js^l\,,\quad \wp^{(i)}(q)y^js^l\,,\qquad 0\le l \le m-1, \quad i,j\ge 0\,,
\end{equation}
among which we have $\Z_m$-invariant elements 
\begin{equation}
    y^j\quad \text{with}\quad j\in m\Z\,,\quad\text{and}\quad 
    \wp^{(i)}(q)y^j\quad \text{with}\quad  i+j+2\in m\Z\,.
\end{equation}
Applying the homomorphism \eqref{res}, we obtain a basis of $B_{\hbar, c}(\E_0)$ of the following form:
\begin{align}\label{wj}
w_j:=\left(\pp -f_{j-1}\right)\dots \left(\pp - f_0\right)\,,&\qquad j\in m\Z\,,
  \\\label{vij}
v^{(i)}_{j}:=
\wp^{(i)}(q)\left(\pp -f_{j-1}\right)\dots \left(\pp - f_0\right)\,,&\qquad i+j+2\in m\Z\,,
\end{align}
where the coefficients $f_k=f_k(q)$ are given by 
\begin{equation}\label{fk}
    f_k=\sum_{l=1}^{m-1}\sum_{x_i\in\mathcal S}c_l(x_i)f(q, x_i)\omega^{-kl}=\sum_{x_i\in\mathcal S}\mu_l(x_i)f(q, x_i)\,. 
\end{equation}

\begin{remark}
The classical algebras $H_{0,c}(\E_0)$ and $B_{0,c}(\E_0)$ are described similarly, by replacing $\pp=\hbar\frac{d}{dq}$ with the classical momentum, $p$ (and by setting $\hbar=0$ in \eqref{rele2}).
\end{remark}

\begin{remark}
    More generally, for any finite $\Z_m$-invariant set $\mathcal Z\subset \E$, consider $U=\E\setminus\mathcal Z$. In that case, the algebras $H_{\hbar, c}(U)$ and $B_{\hbar, c}(U)$ admit a similar description, with the operator $y$ modified as follows:
    \begin{equation}
    y=\hbar\frac{d}{d q}-\sum_{l=1}^{m-1}\sum_{x_i\in U}\sum_{z_j\in\mathcal Z}\frac{1}{|\mathcal Z|} c_l(x_i)f(q-z_j, x_i-z_j)s^l\,.
\end{equation}
The previous case corresponds to $\mathcal Z=\{0\}$.
\end{remark}

\subsection{Hamiltonians}

According to \cite{EFMV11ecm}, the hamiltonians $\hh_1, \dots, \hh_n$ of the elliptic crystallographic Calogero--Moser system for a group $G=\Gamma\rtimes W$ of rank $n$ have the form  
\begin{equation}
\hh_i=f_i(\pp)+\dots\,,\qquad \pp=\left(\hbar\frac{\partial}{\partial q_1}\,,\dots, \hbar\frac{\partial}{\partial q_n}\right)\,,    
\end{equation}
with their leading symbols $f_i$ generating the ring of polynomial $W$-invariants, and with the dots representing terms of smaller order in $\pp$. A similar result holds in the classical case. The connection between these hamiltonians and the elliptic Cherednik algebra is as follows.

\begin{theorem}[\cite{EFMV11ecm}]\label{propham}
The hamiltonians $\hh_i$ represent global sections of the sheaf $B_{\hbar, c}(\C^n/\Gamma, W)$ of spherical Cherednik algebras. Furthermore, they generate the full algebra of global sections, that is, any global section of $B_{\hbar, c}(\C^n/\Gamma, W)$ is a polynomial in $\hh_1, \dots, \hh_n$.
\end{theorem}
The construction of $\hh_i$ given in \cite{EFMV11ecm} is fairly involved, and no explicit expression for $\hh_i$ is known in general. In the rank one case, however, the situation is simpler. In this case, we have a single hamiltonian of the form 
\begin{equation}
 \hh=\pp^m+\dots\,,\qquad \pp=\hbar\frac{d}{dq}\,,   
\end{equation}
which can be described as follows. 
\begin{prop}
    The algebra of global sections of the sheaf of spherical subalgebras $B_{\hbar, c}(\E)$ is generated by a single element of the form 
\begin{equation}
    \label{hamex}
    \hh=w_m+\alpha_2v^{(0)}_{m-2}+\alpha_3v^{(1)}_{m-3}+\dots + \alpha_mv^{(m-2)}_{0}\,.
\end{equation}
Here $w_m$ and $v^{(i)}_{j}$ are the elements defined in \eqref{wj}, \eqref{vij}, and $\alpha_i$ are suitable constant coefficients.
\end{prop}
To prove this, we notice that any global section $\hh$ restricts onto $\E_0=\E\setminus \{0\}$, and so it must be a combinations of elements $w_i$, $v^{(i)}_j$. For $\hh$ to have degree $m$ in $\pp$, it must be obtained from $w_m$ by adding a finite linear combination of $v^{(i)}_j$ with $0\le j \le m-1$. On the other hand, near $q=0$ each global section must belong to the (completion) of the local rational spherical Cherednik subalgebra, generated by $\C[[u]]$, $v$ and $w$ as given in \eqref{abc}. Since $u,v$ are regular at $q=0$, we must have
\begin{equation}\label{pr}
    \hh=(\pp-\mu_{m-1}q^{-1})\dots (\pp-\mu_{0}q^{-1})+\mathrm{regular\  terms}\,,\qquad \mu_j=\mu_j(0).
\end{equation}
Furthermore, each of $w_j, v^{(i)}_j$ near $q=0$ has the following principal part:
\begin{align}
w_j&\sim (\pp+\mmu_{j-1}q^{-1})\dots (\pp+\mmu_{0}q^{-1})\,,
\\
 v^{(i)}_j\,&\sim\, (-1)^i(i+1)!q^{-i-2}
 (\pp+\mmu_{j-1}q^{-1})\dots (\pp+\mmu_{0}q^{-1})\,,\qquad \mmu_l:=\sum_{x_i\in \mathcal S}\mu_l(x_i)\,.
\end{align}
 Comparing this with the previous formula, we conclude that the only allowed terms in $\hh$ are $v^{(i)}_j$ with $j=0,\dots, m-2$ and $i+j=m-2$, thus proving \eqref{hamex}. 

 Now, we may compare the principal parts in \eqref{hamex} and \eqref{pr} to get the relation 
 \begin{multline}
 (\pp-\mu_{m-1}q^{-1})\dots (\pp-\mu_{0}q^{-1})=
 (\pp+\mmu_{m-1}q^{-1})\dots (\pp+\mmu_{0}q^{-1})\\\label{ai}
 +\sum_{i=2}^{m}(-1)^i(i-1)!\alpha_{i} q^{-i}
 (\pp+\mmu_{m-i-1}q^{-1})\dots (\pp+\mmu_{0}q^{-1})\,.
 \end{multline}
This completely determines the coefficients $\alpha_2, \dots, \alpha_m$ entering \eqref{hamex} in terms of 
$\mu_j=\mu_j(0)$ and $\mmu_j=\sum_{x_i\in\mathcal S}\mu_j(x_i)$, i.e. in terms of the parameters of the elliptic Cherednik algebra. In fact, we can trade the parameters $\mu_j(0)$ for $\alpha_j$, that is, regard \eqref{hamex} as depending on $c_l(x_i)_{x_i\in\mathcal S}$ and $\alpha_2, \dots, \alpha_m$. 
 
The classical hamiltonian is described by the same formulas, with $\pp=\hbar\frac{d}{dq}$ replaced by the classical momentum $p$ everywhere.
\begin{prop}
    The algebra of global sections of the sheaf of spherical subalgebras $B_{0, c}(\E)$ is generated by a single element of the form 
\begin{equation}
    \label{hamexc}
    h=w_m+\alpha_2v^{(0)}_{m-2}+\alpha_3v^{(1)}_{m-3}+\dots + \alpha_mv^{(m-2)}_{0}\,.
\end{equation}
Here $w_m$ and $v^{(i)}_{j}$ are the classical analogues of elements \eqref{wj}, \eqref{vij}, 
\begin{align}\label{wjc}
w_j:=\left(p -f_{j-1}\right)\dots \left(p - f_0\right)\,,&\qquad j\in m\Z\,,
  \\\label{vijc}
v^{(i)}_{j}:=
\wp^{(i)}(q)\left(p -f_{j-1}\right)\dots \left(p - f_0\right)\,,&\qquad i+j+2\in m\Z\,,
\end{align}
$f_k$ are given by \eqref{fk}, and $\alpha_i$ are suitable constant coefficients. 
\end{prop}
The relation for determining $\alpha_i$ is obtained by taking the classical limit of \eqref{ai}: 
\begin{multline}
 (p-\mu_{m-1}q^{-1})\dots (p-\mu_{0}q^{-1})=
 (p+\mmu_{m-1}q^{-1})\dots (p+\mmu_{0}q^{-1})\\\label{aic}
 +\sum_{i=2}^{m}(-1)^i(i-1)!\alpha_{i} q^{-i}
 (p+\mmu_{m-i-1}q^{-1})\dots (p+\mmu_{0}q^{-1})\,.
 \end{multline}

\subsection{Explicit formulas}
\label{ellipticform}

Let us write the hamiltonians explicitly, based on the formula \eqref{hamex}. We use the notation $\omega_{1,2,3}$ and $\eta_{1,2}$ for the fixed points, as in Fig.~\ref{table0}. The parameters $\alpha_i$ below are determined through \eqref{ai}.

\paragraph{Case $m=2$:}
In this case, $\tau=\omega_2/\omega_1$ is arbitrary. We denote $g_i:=c_1(\omega_i)$, $i=0\dots 3$; then $\mu_{0}(\omega_i)=-\mu_1(\omega_i)=g_i$.
The hamiltonian  has the form
\begin{equation}\label{hh2}
    \hh=(\pp+f_0)(\pp-f_0)+\alpha_2\wp(q)\,,\qquad f_0=
    \sum_{i=1}^3\frac{g_i\wp'(q)}{2(\wp(q)-e_i)}\,.
\end{equation}
This gives, up to an additive constant,
\begin{equation}
    \hh=\pp\,^2-\sum_{i=0}^3 g_i(g_i-\hbar)\wp(q-\omega_i)\,.
\end{equation}

\paragraph{Case $m=3$:}
 
In this case $\omega_2/\omega_1 = \exp(\pi i/3)$, and we have the parameters $\mu_j(\eta_i)$, $i, j=0,1,2$, where we put $\eta_0=0$ for convenience. 
The hamiltonian  has the form
\begin{align}\label{hh3}
    \hh&=(\pp-f_2)(\pp-f_1)(\pp-f_0)+\alpha_2\wp(q)(\pp-f_0)+\alpha_3\wp'(q)\,,
    \\
    f_j&=
    \sum_{i=1,2}\frac{\mu_j(\eta_i)(\wp'(q)+\wp'(\eta_i))}{2\wp(q)}\,,\qquad j=0,1,2\,.
\end{align}

\paragraph{Case $m=4$:}

In this case $\omega_2/\omega_1 = \exp(\pi i/2)$, and we put $\omega_0=0$ for convenience. We have the parameters $\mu_j(\omega_i)$, $i=0,3$, and $\mu_j(\omega_{1})=\mu_j(\omega_2)$, with $j=0,1,2,3$, and with $\mu_j(\omega_{1,2})=\mu_{j+2}(\omega_{1,2})$. Recall that $\sum_j\mu_j(x_i)=0$ for each fixed point $x_i$.
The hamiltonian has the form
\begin{align} \label{hh4}
   \hh= & (\hat{p} -f_3) (\hat{p} -f_2) (\hat{p} -f_1) (\hat{p} -f_0) \notag \\
     + &  \alpha_2 \wp(q) (\hat{p} -f_1) (\hat{p} -f_0) 
      \notag \\
     + & \alpha_3 \wp'(q) (\hat{p} -f_0) \notag \\
     + &  \alpha_4 \wp''(q)\,,
\end{align}
with 
\begin{align}
    f_j =  4\mu_j(\omega_{1,2})\frac{\wp(q)^2}{ \wp'(q)}  + \frac{1}{2} \mu_j(\omega_3) \frac{\wp'(q)}{ \wp(q)}\,. 
\end{align}

\paragraph{Case $m = 6$:}

In this case $\omega_2/\omega_1=e^{\pi i/3}$, and we have six parameters $\mu_j(0)$, $j=0,\dots, 5$, three parameters 
$\mu_j(\eta_{1})=\mu_j(\eta_2)$, $j=0,1,2$, and two parameters $\mu_j(\omega_{1})=\mu_j(\omega_2)=\mu_j(\omega_{3})$, $j=0,1$. Recall that $\sum_j\mu_j(x_i)=0$ for each fixed point. Also, we extend $\mu_j(x_i)$ by $\mu_j(\eta_{1,2})=\mu_{j+3}(\eta_{1,2})$ and $\mu_j(\omega_{1,2,3})=\mu_{j+2}(\omega_{1,2,3})$.
The hamiltonian has the form
\begin{align} \label{hh6}
  \hh= & (\hat{p} -f_5)(\hat{p} -f_4)(\hat{p} -f_3) (\hat{p} -f_2) (\hat{p} -f_1) (\hat{p} -f_0) \notag \\
     + &  \alpha_2 \wp(q) (\hat{p} -f_3) (\hat{p} -f_2) (\hat{p} -f_1) (\hat{p} -f_0) \notag \\
     + &  \alpha_3 \wp'(q) (\hat{p} -f_2) (\hat{p} -f_1) (\hat{p} -f_0) \notag \\
     + &  \alpha_4 \wp''(q)  (\hat{p} -f_1) (\hat{p} -f_0) \notag \\
     + &  \alpha_5 \wp^{(3)}(q)  (\hat{p} -f_0) \notag \\
     + &  \alpha_6 \wp^{(4)}(q)\,.
\end{align}
Here
the coefficients $f_j$ are found to be 
\begin{align}
    f_j =  \mu_j(\omega_{1,2,3}) \frac{6\wp(q)^2}{ \wp'(q)}
    + \mu_j(\eta_{1,2}) \frac{\wp'(q)}{ \wp(q)}. 
\end{align}

\medskip 
In all cases, the expression for $\hh$ can be expanded into the form
    $\hh=\pp^{\,m}+A_2\,\pp^{\,m-2}+\dots +A_m$\,,
see Appendix \ref{a:elliptic}.

\subsection{Rational form} \label{rf}
We can further convert $\hh$ into a rational form by using the $\Z_m$-invariant coordinate $x=u(q)$ in accordance with Table \ref{table1}. Then 
\begin{equation}
    \frac{d}{dq}=w\frac{d}{dx}\,,\quad\text{with }\ w:=\frac{du}{dq}\,.
\end{equation}
By straightforward manipulations, the expression \eqref{hamex} after rearranging leads to 
\begin{equation}
    w^{-m}\hh=D_{m-1}\dots D_0+\sum_{j=2}^m A_jD_{m-j-1}\dots D_0\,,
\end{equation} 
where 
\begin{equation}
    D_j:=w^{-j-1}(\pp-f_j)w^j=\hbar\frac{d}{dx}-\frac{f_j}{w}+j\hbar\frac{w'}{w^2}\,, \qquad A_j:=\alpha_j\frac{\wp^{(j-2)}(q)}{w^j}\,.
\end{equation}
Let us write explicit expressions for each of $m=2,3,4,6$.

\paragraph{Case $m=2$:}

In this case, $x=\wp(q)$ and $w=\wp'(q)$, with $w^2=4(x-e_1)(x-e_2)(x-e_3)$, $e_i=\wp(\omega_i)$. From \eqref{hh2}, we get
\begin{equation}
    w^{-2}\hh=\left(\hbar\frac{d}{dx}+\sum_{i=1,2,3}\frac{g_i+\hbar}{2(x-e_i)}\right)
     \left(\hbar\frac{d}{dx}-\sum_{i=1,2,3}\frac{g_i}{2(x-e_i)}\right)+\frac{\alpha_2 x}{4(x-e_1)(x-e_2)(x-e_3)}\,.
\end{equation}

\paragraph{Case $m=3$:}

In this case, $x=\frac12\wp'(q)$, $w=\frac12\wp''(q)=3\wp^2(q)$, and $w^3=3^3(x-e_1)^2(x-e_2)^2$ where $e_i=\frac12\wp'(\eta_i)$. We have parameters $\mu_j(\eta_i)$, $i,j=0,1,2$. From \eqref{hh3}, we obtain
\begin{equation}
  w^{-3}\hh = D_2D_1D_0+\frac{\alpha_2}{3^2(x-e_1)(x-e_2)}D_0+\frac{2\alpha_3x}{3^3(x-e_1)^2(x-e_2)^2}\,,
\end{equation}
where 
\begin{equation}
    D_j=\hbar\frac{d}{dx}-\sum_{i=1,2}\frac{\mu_j(\eta_i)-2j\hbar}{3(x-e_i)}\,.
\end{equation}

\paragraph{Case $m=4$:}

In this case, $x=\wp^2(q)$, $w=2\wp(q)\wp'(q)$, and $w^4=4^4(x-e_1)^3(x-e_2)^2$
where $e_1=\wp^2(\omega_3)=0$ and $e_2=\wp^2(\omega_{1,2})$. By a linear transformation in $x$, we can make $e_1, e_2$ arbitrary. We have parameters $\mu_j(\omega_3)$ and $\mu_j(\omega_1)=\mu_j(\omega_2)$ for $j=0,1,2,3$, with the property $\mu_j(\omega_{1,2})=\mu_{j+2}(\omega_{1,2})$. From \eqref{hh4} we obtain
\begin{align}
  w^{-4}\hh = & D_3 D_2 D_1 D_0
     +   \frac{\alpha_2}{4^2 \left(x-e_1\right) \left(x-e_2\right)}  D_1 D_0\notag \\
     + &  \frac{2\alpha_3}{4^3 \left(x-e_1\right)^2 \left(x-e_2\right)} D_0 
     +   \frac{2\alpha_4(3x-2e_1-e_2)}{4^4 \left(x-e_1\right)^3 \left(x-e_2\right)^2}  \,,   
\end{align}
where 
\begin{align}
    D_j&=\hbar\frac{d}{dx}-\frac{\mu_j(\omega_3)-3j\hbar}{4(x-e_1)}-\frac{\mu_j(\omega_{1,2})-j\hbar}{2(x-e_2)}\,.
    \end{align}

\paragraph{Case $m=6$:}
We use $x=\wp^3(q)$ and $w=3\wp^2(q)\wp'(q)$, so $w^6=6^6(x-e_1)^4(x-e_2)^3$. Here $e_1=0$ but we may make it arbitrary by a shift in the $x$-variable. 
We have parameters $\mu_j(0)$, $\mu_j(\omega_{1,2,3})$ and $\mu_j(\eta_{1,2})$. Let us use the shorthand $\mu_j^{(\omega)}:=\mu_j(\omega_{1,2,3})$ and $\mu_j^{(\eta)}:=\mu_j(\eta_{1,2})$. These have the periodicity property $\mu_j^{(\omega)}=\mu_{j+2}^{(\omega)}$, $\mu_j^{(\eta)}=\mu_{j+3}^{(\eta)}$. 
From \eqref{hh6} we obtain
\begin{align}
     w^{-6}\hh =   & D_5 D_4 D_3 D_2 D_1 D_0
     +   \frac{\alpha_2}{6^2 \left(x-e_1\right) \left(x-e_2\right)} D_3 D_2 D_1 D_0\notag \\
     + &  \frac{2\alpha_3}{6^3 \left(x-e_1\right)^2 \left(x-e_2\right)} D_2 D_1 D_0 \notag 
     +   \frac{6\alpha_4}{6^4 \left(x-e_1\right)^2 \left(x-e_2\right)^2} D_1 D_0 \notag \\
     + &  \frac{24\alpha_5}{6^5 \left(x-e_1\right)^3 \left(x-e_2\right)^2} D_0 
     +   \frac{24\alpha_6(5 x  -3 e_1-2 e_2)}{6^6 \left(x-e_1\right)^4 \left(x-e_2\right)^3}\,,
\end{align}
where
\begin{align}
    D_j = \hbar \frac{d}{dx}-\frac{{\mu}_j^{(\eta)}- 2 j \hbar}{3(x-e_1)} -\frac{{\mu}_j^{(\omega)}- j \hbar}{2(x-e_2)}\,.
\end{align}

\section{Classical dynamics, Dunkl operator, Lax matrix, and the spectral curve}\label{sec3}

In this section we look at the dynamics of the hamiltonian \eqref{hamexc}, its Lax presentation, and the geometry of the spectral curves. 

\subsection{Classical dynamics} 
Let $h$ be one of the hamiltonians \eqref{hamexc} for $m=2,3,4,6$. The corresponding dynamics is described by the Hamilton--Jacobi equations,
\begin{equation}\label{hj}
    \frac{dp}{dt}=-\frac{\partial h}{\partial q}\,,\qquad \frac{dq}{dt}=\frac{\partial h}{\partial p}\,.
\end{equation}
The motion takes place along the level curves
\begin{equation}\label{fib}
    \Ssigma=\{(p, q)\,:\,h(p, q)=z \}\,.
\end{equation}
Each curve is an $m$-sheeted branched covering of the elliptic curve $\E$. The curves are not compact, due to $h$ having poles at $q=x_i$, and have a fairly high genus. To interpret \eqref{hj} as a {\it complex integrable system}, one needs to compactify the curves and take into account the $\Z_m$-symmetry
\begin{equation}\label{ace}
    s\,:\ p\to\omega^{-1}p\,,\quad q\to\omega q\,,\qquad \omega=e^{2\pi i/m}\,.
\end{equation}
This is summarised in the next proposition. 

\begin{prop}\label{prop5} $(1)$ Suppose $x_i\in\E$ is a fixed point for $W=\Z_m$, with the stabiliser $\Z_{m_i}$. The $m$-sheeted branched covering $\Ssigma\to\E$, $(p,q)\mapsto q$ near $q=x_i$ has the form
\begin{equation}\label{locsh}
    \prod_{j=0}^{m-1}\left(p-\frac{\mu_j}{q-x_i}+{O}\left((q-x_i)^{m_i-1}\right)\right)=0\,.
\end{equation}
Here $\mu_j=\mu_j(x_i)$ are the ``linear masses'' \eqref{emu}.

$(2)$ A compactification of $\Ssigma$ is obtained by adding $m$ distinct points over each fixed point $x_i$, one point for each of the sheets \eqref{locsh}. For generic couplings $c$, the compactified curve is smooth, of genus $g=m^2+1$. The $\Z_m$-action \eqref{ace} is free on $\Ssigma$ and has the stabiliser $\Z_{m_i}$ for the points above $q=x_i$.

$(3)$ The (compactified) quotient curves $\Sigma:=\Ssigma/\Z_m$ have genus one. The differential 
\begin{equation}
    dt=\frac{dp}{-\frac{\partial h}{\partial q}}=\frac{dq}{\frac{\partial h}{\partial p}}
\end{equation}
defines a non-vanishing holomorphic $1$-form on $\Sigma$. Hence, the $\Z_m$-quotient of the fibration \eqref{fib} defines an elliptic fibration on $T^*\E/\Z_m$, and the dynamics \eqref{hj} becomes linear along the fibers.    
\end{prop}
It is possible to prove this proposition by analysing the formula \eqref{hamexc}. However, we will use an alternative method and derive it from a {\it Lax presentation} for the system \eqref{hj}. Such a Lax presentation can be constructed following \cite{Chalykh19qlp} by using Dunkl operators, which we need to introduce first.

\subsection{Elliptic Dunkl operators}
Elliptic analogues of Dunkl operators go back to \cite{BFV94edo}, see also \cite{Cherednik95eqm}. For general complex crystallographic groups, they were introduced in \cite{EM08edo}. Like their rational counterparts, elliptic Dunkl operators form a commutative family, but their symmetry properties are more complicated due to the presence of auxiliary ``spectral'' variables. Below we discuss the case of rank $1$ only, in which case we need just one Dunkl operator. We follow the general framework of \cite{EM08edo, EFMV11ecm}, but since we are dealing with a special case, everything will be made very concrete.        

Let us fix some notation. For an abelian variety $X$ with an action of a finite group $W$, the Dunkl operators \cite{EM08edo} depend on an auxiliary variable $\alpha\in X^\vee=\mathrm{Pic}_0(X)$. We will be dealing with the case $X=\E=\C/\L$ and $W=\Z_m$, in which case we identify $X^\vee\simeq X$ so that $\alpha\in\E$. Below $\sigma$, $\zeta$, $\wp$ stand for the Weierstrass functions associated to the lattice $\L=2\omega_1\Z+2\omega_2\Z$, and we write $\varphi(x,z)$ for the following combination: 
\begin{align}
    \varphi(x,z) = \frac{\sigma(x - z )}{\sigma(x) \sigma(-z) }.
\end{align}
Recall that for each fixed point $x_i\in \E$ we have parameters $c_l(x_i)$, $1\le l\le m-1$, with $c_l(x_i)=0$ unless $s^l(x_i)=x_i$. The fixed points of $s^l$ are identified with the cosets
\begin{equation}
    (\Omega_l)^{-1}\L/\L\,,\qquad \Omega_l:=1-\omega^l\,.
\end{equation}
Now introduce the following functions $v_l=v_{l, c}$ of $x,z\in\C$:
\begin{align}\label{v}
v_{l}( x,z) & =\sum _{\{x_{i}\}} c_{l}( x_{i}) e^{-\eta ( \Omega_{l} x_{i}) z} \varphi ( x-x_{i} ,\Omega_{-l} z)\,,\qquad l\in\Z_m\setminus\{0\}\,,
\end{align}
where the summation is over fixed points $x_i\in (\Omega_l)^{-1}\L/\L$, and $\eta(\gamma)$ for $\gamma=2n_1\omega_1+2n_2\omega_2\in \L$ is defined by
\begin{equation}\label{eta}
    \eta(\gamma)=-\int_q^{q+\gamma}\wp(q)\,dq=\zeta(q+\gamma)-\zeta(q)=2n_1\zeta(\omega_1)+2n_2\zeta(\omega_2)\,.
\end{equation}
We extend $\eta$ by $\R$-linearity so that $\eta(2a\omega_1+2b\omega_2)=2a\zeta(\omega_1)+2b\zeta(\omega_2)$ for $a,b\in\R$. The symmetry of the lattice $\L$ implies that $\eta(\omega\gamma)=\omega^{-1}\eta(\gamma)$. Note that the formula \eqref{v} is invariant under $x_i\mapsto x_i+\gamma$, $\gamma\in \L$, thus independent on the choice of coset representatives $x_i\in(\Omega_l)^{-1}\L/\L$. An important property of the functions $v_l(x,z)$ is the {\it duality},
\begin{equation}\label{dudu}
 v_{l,c}(x,z)=-v_{-l,c^{\vee}}(z,x)\,,   
\end{equation}
where the set of {\it dual couplings} $c^\vee$ is described in Appendix \ref{a:duality}.

\bigskip 

We now define the {\it elliptic Dunkl operator} for $W=\Z_m$ by the formula
\begin{align}\label{due}
   y = \hbar \frac{d}{dq } - \sum_{ l = 1}^{m-1} v_{l}( q , \alpha) s^l\,, 
\end{align}
with the coefficients $v_l(x,z)$ given by \eqref{v}. One important feature of this case is that $y$ does {\it not} belong to the elliptic Cherednik algebra $H_{\hbar, c}(\E)$, since the coefficients $v_l(q, \alpha)$ are not elliptic. Another distinctive feature is the dependence on the auxiliary ``spectral'' variable $\alpha$; we write $y=y(\alpha)$ to indicate that dependence. As a function of $\alpha$, the Dunkl operator has simple poles at the fixed points $\alpha=x_i$, and it has the following properties:
\begin{align}\label{proa}
    &y(\alpha+\gamma)
        =e^{-\eta(\gamma)q}\,y(\alpha)\,e^{\eta(\gamma)q}-\hbar\eta(\gamma)\,\,\quad\forall\ \gamma\in \L\,,
    \\\label{prob}
    &\mathrm{res}_{\alpha=x_i} y(\alpha)= \sum_{l=1}^{m-1}e^{-\eta(\Omega_{-l}x_i)q}
    c_{-l}^\vee(x_i)s^l\,,
    \\\label{proc}
    &s\,y(\alpha)=\omega^{-1}y(\omega^{-1}\alpha)\,s\,.
\end{align}
The formula \eqref{prob} follows from the obvious property $\mathrm{res}_{x=x_i}v_l(x,z)=c_l(x_i)e^{-\eta(\Omega_{l}x_i)z}$ and
the duality \eqref{dudu}. Since ${\eta(\Omega_{-l}x_i)q}=\eta(x_i)(1-\omega^l)q$, the relation \eqref{prob} can be rewritten as
\begin{equation}\label{probb}
 \mathrm{res}_{\alpha=x_i} y(\alpha)= e^{-\eta(x_i)q}\left(\sum_{l=1}^{m-1}
    c_{-l}^\vee(x_i)s^l \right)e^{\eta(x_i)q}\,.   
\end{equation}
The classical Dunkl operator is defined as
\begin{align}\label{duec}
   y^c = p - \sum_{ l = 1}^{m-1} v_{l}( q , \alpha) s^l\,. 
\end{align}
It has the same properties as in the quantum case, namely, \eqref{proa} (with $\hbar=0$), \eqref{proc} and \eqref{probb}. 

\medskip

\begin{remark}\label{pic}
Let $\mathcal L_\alpha$ denote the line bundle over $\E$ given by the quotient $(\C \times \mathbb{C})/\sim$ with $(q, \xi) \sim (q + \gamma, \exp(-\eta(\gamma)\alpha)\xi)$ for $\gamma\in \L$. Under the $\Z_m$-action on $\E$, we have $(\mathcal L_\alpha)^s=\mathcal L_{\omega^{-1}\alpha}$. The function $v_l(q,\alpha)$, for a fixed $\alpha$, represents a meromorphic section of $\mathcal L_{\Omega_{-l}\alpha}\simeq \mathcal{L}_\alpha \otimes (\mathcal{L}^{s^l}_\alpha)^\ast$. Hence, the Dunkl operator $y(\alpha)$ acts on (meromorphic) sections of $\mathcal L_\alpha$ in agreement with conventions in \cite{EM08edo}. 
\end{remark}

\subsection{Lax matrix}
We will make use of the fact that the system \eqref{hj} admits a Lax presentation
\begin{equation}\label{laxe}
    \frac{dL}{dt}=[L,A]\,,
\end{equation}
for suitable matrices $L=L(p,q)$, $A=A(p,q)$. The Lax pair $L, A$ can be found following the method of \cite{Chalykh19qlp}. We only need the Lax matrix $L$; it is calculated from the Dunkl operator \eqref{duec} by adapting the recipe from \cite{Chalykh19qlp}. Namely, let $\C(p,q)$ denote the space of (meromorphic) functions of $p,q\in\C$,
with the $\Z_m$-action $s(p,q)=(\omega^{-1}p, \omega q)$. One considers $\C(p,q)\rtimes \Z_m$ acting on itself by left multiplication. If we use a vector-space isomorphism $\C(p,q)\rtimes \Z_m\simeq \C\Z_m\otimes \C(p,q)$, we can interpret the action of any element as a $\Z_m\times \Z_m$ matrix with entries from $\C(p,q)$. For example, multiplication by $q$ and $s$ are represented by the following matrices $Q$ and $S$, respectively:
\begin{align}\label{s}
 Q=\mathrm{diag}(q,\omega^{-1}q,\dots, \omega^{-m+1}q),,\quad S=\sum_{i\in\Z_m} E_{i+1,i}=
    \begin{pmatrix}
    0&\dots&\dots&0&1\\
    1&0&\dots&\dots&0\\
    0&1&0&\dots&0\\
    &&\ddots&&\\
    0&\dots&0&1&0
    \end{pmatrix}\,.   
\end{align}
The action of the Dunkl operator \eqref{duec} is then represented by the following {\it Lax matrix} $L=(L_{ij})$:
\begin{equation}\label{lax}
    L_{ij}=\begin{cases}
    \omega^ip\quad &\text{for}\ i=j\,,\\
    v_{i-j}(\omega^{-i}q, \alpha)\quad &\text{for}\ i\ne j\,,
\end{cases}
\qquad (i,j\in\Z_m)\,.
\end{equation}
We write $L=L(\alpha)$ to indicate the dependence on the spectral parameter. The Lax matrix has first order poles at the fixed points $\alpha=x_i$ and it has the following properties: 
\begin{align}\label{laxpr1}
    &L(\alpha+\gamma)=e^{-\eta(\gamma)Q}L(\alpha)e^{\eta(\gamma)Q}\quad\forall\ \gamma\in\L\,,
    \\
    \label{laxpr2}
    &\mathrm{res}_{\alpha=x_i} L(\alpha)= e^{-\eta(x_i)Q}\left(\sum_{l=1}^{m-1}
    c_{-l}^\vee(x_i)S^l \right)e^{\eta(x_i)Q}\,,
    \\
    \label{laxpr3}
    &\,L(\omega^{-1}\alpha)=\omega SL(\alpha)\,S^{-1}\,,
\end{align}
with the above matrices $Q$ and $S$. These properties immediately follow from \eqref{proa} (with $\hbar=0$), \eqref{proc} and \eqref{probb}; alternatively, they can be verified directly.

The method of \cite{Chalykh19qlp} proves that the above $L$ admits a Lax partner $A$ so that the equation \eqref{laxe} holds (see Remark \ref{laxA} below). Hence, the coefficients $b_i$ of the characteristic polynomial 
\begin{equation}
    \det(L-k\Id)=(-1)^m\left(k^m+b_1k^{m-1}+\dots+b_m\right)
\end{equation}
remain constant under the hamiltonian dynamics \eqref{hj}. Note that the hamiltonian $h(p,q)$ has degree $m$ in $p$. Since the coefficients $b_i=b_i(\alpha; p,q)$ have degree $<m$ in $p$ if $i<m$, we must have $b_i=b_i(\alpha)$. On the other hand,
\begin{equation}
    b_m=(-1)^m\det L= (-1)^m(\prod_{i=0}^{m-1}\omega^i)p^m+\dots=-p^m+\dots\,,
\end{equation}
so we must have $b_m=-h(p,q)+b_m(\alpha)$. As a result,
\begin{equation}\label{las}
    \det(L-k\Id)=(-1)^m\left(k^m+b_1(\alpha)k^{m-1}+\dots+b_m(\alpha)-h(p,q)\right)\,. 
\end{equation}
To find an explicit formula for $\det(L-k\Id)$, one may try calculating the determinant directly, but this seems daunting for $m=4,6$. Instead, we will obtain the answer momentarily from the following symmetry of $L$. Namely, write $L=L(p,q;\alpha)$ for the Lax matrix \eqref{lax}, and $L^\vee$ for the Lax matrix with the dual couplings $c^\vee$. Then we have the following relation:
\begin{equation}\label{dulax}
    \det(L(p,q;\alpha)-k\Id)=-\det(L^\vee(k,\alpha; q)-p\Id)\,.
\end{equation}
Indeed, the matrix in the r.h.s has the following entries:
\begin{equation}
    (L^\vee-p\Id)_{ij}=\begin{cases}
    \omega^ik-p=-\omega^i(w^{-i}p-k)\quad &\text{for}\ i=j\,,\\
    v_{i-j, c^\vee}(\omega^{-i}\alpha, q)=-\omega^iv_{j-i, c}(\omega^iq, \alpha)\quad &\text{for}\ i\ne j\,.
\end{cases}
\end{equation}
Thus, $(L^\vee-p\Id)_{ij}=-\omega^i(L-k\Id)_{m-i, m-j}$ and
\begin{equation}\label{dul}
    L^\vee-p\Id=-\mathrm{diag}(1,\omega,\dots, \omega^{m-1}) C(L-k\Id)C^{-1}\,,\qquad C=\sum_{i\in\Z_m} E_{i,-i}\,,
\end{equation}
which makes \eqref{dulax} obvious. Now, combining \eqref{las} and \eqref{dulax}, we get
\begin{equation}\label{spc}
   (-1)^m \det(L-k\Id)=h^\vee(k,\alpha)-h(p,q)\,,
\end{equation}
where $h^\vee$ denotes the hamiltonian \eqref{hamex} with the dual couplings $c^\vee$.

\begin{remark}\label{laxA}
Following \cite{Chalykh19qlp}, let us substitute the classical Dunkl operator \eqref{duec} into the classical dual hamiltonian $h^\vee$.  By \cite[(5.19)]{Chalykh19qlp}, this recovers the classical hamiltonian $h(p,q)$, i.e.,
\begin{equation}
h^\vee(y^c,\alpha)=h(p,q)\,.    
\end{equation}
In the representation discussed above, $y^c$ becomes the Lax matrix $L$ and this relation turns into 
\begin{equation}
h^\vee(L,\alpha)=h(p,q)\Id\,.    
\end{equation}
This gives another proof of \eqref{spc}. Furthermore, if one uses instead the {\it quantum} Dunkl operator $y$, then according to \cite[(5.20)]{Chalykh19qlp} one gets
\begin{equation}
h^\vee(y,\alpha)=\hh+\widehat{a}\,,    
\end{equation}
for a suitable $\widehat a\in \mathcal D_\E\rtimes \Z_m$. The classical limit of $\hbar^{-1}\widehat a$ as $\hbar\to 0$ gives an element $a\in\C(p,q)\rtimes \Z_m$. As explained in \cite{Chalykh19qlp}, the matrix $A=A(p,q)$ representing $a$ gives a Lax partner for $L$, satisfying \eqref{laxe}.  
\end{remark}

\subsection{Spectral curves}
The formula \eqref{spc} gives us an explicit one-parameter family of the {\it spectral curves} $\det(L-k\Id)=0$ as
\begin{equation}\label{specd}
    \Ssigma^\vee=\{(k,\alpha)\,:\,h^\vee(k, \alpha)=z \}\,,
\end{equation}
parameterised by the value $z$ of the hamiltonian $h(p,q)$. The Lax matrix $L$ has $m$ distinct eigenvalues generically (this is obviously true if $c=0$, hence also true for generic couplings). When $\alpha$ approaches a fixed point $\alpha=x_i$, the eigenvalues tend to infinity, and their behaviour is determined by $\mathrm{res}_{\alpha=x_i}L$. Using \eqref{laxpr2}, we find that the eigenvalues of $L$ near $\alpha=x_i$ are given by
\begin{equation}\label{brak}
    k=\frac{\mu_j^\vee}{\alpha-x_i}+O(1)\,,\qquad \mu_j^\vee=\mu_j^\vee(x_i):=\sum_{l=1}^{m-1}\omega^{jl}c_{-l}^\vee(x_i)\,,\qquad j=0,\dots, m-1\,.
\end{equation}
If the stabiliser of $x_i$ is $\Z_{m_i}\subset \Z_m$, then $c^\vee_l=0$ unless $lm_i$ is zero modulo $m$. As a result, $\mu_j^\vee=\mu_{j+m_i}^\vee$, so among $\mu_j^\vee$ there will be only $m_i$ different values. For example, for $m=6$ and a fixed point with stabiliser $\Z_3$, we have 
\begin{equation}
  \{\mu_j^\vee\}=(\mu_0^\vee,\mu_1^\vee, \mu_2^\vee,\mu_0^\vee,\mu_1^\vee,\mu_2^\vee)\,,\qquad \mu_0^\vee+\mu_1^\vee+\mu_2^\vee=0\,.  
\end{equation}

As is readily seen from \eqref{laxpr3}, the spectral curve $\Ssigma^\vee$ is invariant under the $\Z_m$-action
\begin{equation}\label{aced}
    s\,:\ k\to\omega k\,,\quad \alpha\to\omega^{-1}\alpha\,,\qquad \omega=e^{2\pi i/m}\,.
\end{equation}
(This also follows from the $\Z_m$-invariance of the hamiltonian.) Let
\begin{equation}
    \Sigma^\vee:=\Ssigma^\vee/\Z_m
\end{equation}
be the quotient curve. Both curves may be viewed as $m$-sheeted branched coverings $\Ssigma^\vee\to \E$ and $\Sigma^\vee\to \E/\Z_m=\mathbb{P}^1$, respectively. The action of the stabiliser $\Z_{m_i}$ does not permute the sheets \eqref{brak}: indeed, the sheets near $\alpha=0$ have distinct $\mu_j$'s so cannot be permuted. As a result, each sheet is invariant under the stabiliser $\Z_{m_i}$ of $\alpha=x_i$, which implies that the $O(1)$ term in \eqref{brak} must have correct symmetry, hence
\begin{equation}\label{brak1}
    k=\frac{\mu_j^\vee}{\alpha-x_i}+O\left((\alpha-x_i)^{m_i-1}\right)\,.
\end{equation}
This shows that $\Ssigma^\vee$ can be compactified by adding $m$ points above each $\alpha=x_i$, so that the $\Z_m$-action extends to the compactification and the added points have $\Z_{m_i}$ as their stabilisers. It also implies that in local invariant coordinates $\epsilon=(\alpha-x_i)^{m_i}$ and $s=(\alpha-x_i)k$, each branch \eqref{brak1} becomes 
\begin{equation}
    s=\mu_j^\vee+\epsilon\, r_j(\epsilon)\,,\quad \text{with some}\ r_j\in \C[[\epsilon]]\,.
\end{equation}
This means that locally around $\epsilon=0$, the branched covering $\Sigma^\vee \to \mathbb P^1$ is {\it unramified}.

\subsection{Elliptic fibration, integrability, and Seiberg--Witten differential}

The above analysis of the curves $h^\vee(k,\alpha)=z$, immediately carries over to the fibration \eqref{fib}: one just needs to replace $c^\vee$ by $c$. This partly proves Proposition \ref{prop5}. To prove the remaining claims, let us proceed by calculating the genus of $\Ssigma^\vee$ and $\Sigma^\vee$. There are various ways to do that, and we choose the one which is best for our purposes.  

Consider the following meromorphic differentials on $T^*\E$: 
\begin{equation}
    \Omega_1=\frac{dk}{-{\partial h^\vee}/{\partial \alpha}}\,,\quad \Omega_2=\frac{d\alpha}{{\partial h^\vee}/{\partial k}}\,.
\end{equation}
Obviously, $\Omega_1=\Omega_2$ on the level curves \eqref{specd}, and the resulting $1$-form $\Omega:=\Omega_1=\Omega_2$ is holomorphic and non-vanishing on $\Ssigma^\vee$ away from the fixed points of $\Z_m$. To analyse it near  $\alpha=x_i$, we take $\alpha$ as a local coordinate and use that  
\begin{equation}
    {{\partial h^\vee}/{\partial k}}=\frac{\partial f(k, \alpha)}{\partial k}\,,\qquad f(k, \alpha)=(-1)^m\det(L-k\Id)\,,
\end{equation}
with 
\begin{equation}\label{fka}
    f(k,\alpha)=\prod_{j=0}^{m-1}\left(k-\frac{\mu_j^\vee}{\alpha-x_i}+\beta_j(\alpha-x_i)^{m_i-1}+\dots\right)\,.
\end{equation}
Picking one of the local branches \eqref{brak1}, we see that $m/m_i$ factors in \eqref{fka} behave as $(\alpha-x_i)^{m_i-1}$, while the remaining $m-m/m_i$ factors behave as $(\alpha-x_i)^{-1}$. Differentiating $f$ with respect to $k$ removes one of the factors; from this,
\begin{equation}
  \frac{\partial f(k, \alpha)}{\partial k}\sim (\alpha-x_i)^{(m_i-1)(m/m_i-1)-(m-m/m_i)}=(\alpha-x_i)^{-m_i+1}\,.  
\end{equation}
As a result, on each branch, 
\begin{equation}
   \Omega=\frac{d\alpha}{{\partial h^\vee}/{\partial k}}\sim (\alpha-x_i)^{m_i-1}d\alpha\,.
\end{equation}
Hence, $\Omega$ is holomorphic on $\Ssigma^\vee$, and so the overall number of its zeros over $\alpha=x_i$ is $m(m_i-1)$. There are $m/m_i$ fixed points in the $\Z_m$-orbit of $x_i$, so they contribute $m^2(m_i-1)/m_i$ zeros. The total number of zeros is therefore  
\begin{equation}
   m^2 \sum_i \left(1-\frac{1}{m_i}\right)\,.
\end{equation}
Here $(m_i)=(2,2,2,2)$ for $m=2$, $(m_i)=(3,3,3)$ for $m=3$, $(m_i)=(4,4,2)$ for $m=4$, and $(m_i)=(6,3,2)$ for $m=6$. In all cases, the above sum gives $2m^2$, so $\Omega$ is a holomorphic differential on $\Ssigma^\vee$ with $2m^2$ zeros. From that, the genus of $\Ssigma^\vee$ is $m^2+1$. Next, the form $\Omega$ is clearly $\Z_m$-invariant so it defines a holomorphic form on $\Sigma^\vee$. Near $\alpha=x_i$ we use $x:=(\alpha-x_i)^{m_i}$ as a local coordinate; then 
\begin{equation}
   \Omega\sim (\alpha-x_i)^{m_i-1}d\alpha\sim dx\,.
\end{equation}
Hence, $\Omega$ is non-vanishing on $\Sigma^\vee$, so $\Sigma^\vee$ has genus one. This establishes all the remaining claims in Proposition \ref{prop5}.

\begin{remark}
When viewed on $\Ssigma^\vee$, $\Omega$ is holomorphic and $\Z_m$-invariant. Up to a factor, there is only one such 1-form. Indeed, by local symmetry, it must have zero of order at least $m_i-1$ at each point with stabiliser $\Z_{m_i}$. Thus, it is bound to have the same divisor as $\Omega$.     
\end{remark}

We finish the section by exhibiting a Seiberg--Witten differential for the elliptic fibration on $T^*\E/\Z_m$. The canonical holomorphic symplectic form on $T^*\E$ is $\omega=dk\wedge d\alpha=d\lambda$, for $\lambda$ the canonical Liouville 1-form,
\begin{equation}
    \lambda=k\,d\alpha\,.
\end{equation}
Both $\omega$ and $\lambda$ are $\Z_m$-invariant so descend to holomorphic forms on $T^*\E/\Z_m$. On the compactified fibers the form $\lambda$ is only meromorphic; the reason being that on any particular branch near $\alpha=x_i$ we have 
\begin{equation}\label{swres}
    \lambda=\left(\frac{\mu_j^\vee}{\alpha-x_i}+O\left((\alpha-x_i)^{m_i-1}\right)\right)d\alpha\,.
\end{equation}
We therefore conclude that $\lambda$ has only simple poles and {\it constant} (i.e. independent of $z$) residues $\mu_j^\vee=\mu_j^\vee(x_i)$. This allows us to view $\lambda$ as a {\it Seiberg--Witten (SW) differential} for the elliptic fibration on $T^*\E/\Z_m$. The residues of the SW differential are referred to as {\it linear masses}; as we see, they are directly related to the coupling parameters of the integrable system.  

\section{Spectral curves and elliptic pencils}\label{pencils}

To interpret the elliptic fibration 
\begin{equation}\label{pol1}
    h^\vee(k,\alpha)=z
\end{equation}
geometrically, we convert it into a polynomial form, using the $\Z_m$-invariant combinations
\begin{equation}
    x=u(\alpha)\,,\qquad y=v(\alpha)k\,,
\end{equation}
where $u, v$ are the functions from Table \ref{table1}. Equally, we may consider the fibration by the level sets of the hamiltonian $h$,
\begin{equation}\label{pol0}
    h(p,q)=z\,,
\end{equation}
 writing it in terms of $x=u(q)$ and $y=v(q)p$. The only difference between the two fibrations is in the coupling parameters: $c^\vee$ or $c$, respectively. The polynomial form of the fibration \eqref{pol0} is presented in Appendix \ref{a:pencils}. It is of the form
\begin{equation}\label{pen0}
    y^m+\sum_{j=2}^m Q_j(x)y^{m-j}=zP_m(x)\,,
\end{equation}
where $P_m(x)$ is the same as in \eqref{mcov}, \eqref{cci2}--\eqref{cci6}.

The SW differential in terms of $x,y$ is chosen as
\begin{equation}\label{sw1}
    \lambda=\frac{y\,dx}{(x-e_1)(x-e_2)(x-e_3)}\,\quad(m=2)\,,\qquad \lambda=\frac{y\,dx}{(x-e_1)(x-e_2)}\,\quad(m=3,4,6)\,.
\end{equation}
As we verify in Appendix \ref{a:pencils}, the fibration \eqref{pen0} describes an elliptic pencil of a special form. Below we describe it geometrically. 

We first consider the cases of $m=3,4,6$. It will be convenient to work in homogeneous coordinates in the projective plane, replacing \eqref{pen0} with
\begin{equation}\label{penh}
    Q(x,y,w)-zP(x,w)=0\,,
\end{equation}
where $(x:y:w)$ are the homogeneous coordinates on $\mathbb{P}^2$, and $z$ parameterises the pencil. The base of each pencil will be a collection of points on a union of three lines $\ell_0$, $\ell_1$, $\ell_2$ meeting at a point. Up to projective equivalence, we can always assume that the lines are  
\begin{equation}\label{lines}    \ell_0:\  w=0\,,\quad\ell_1:\ x-e_1w=0\,,\qquad\ell_2:\ x-e_2w=0\,. \end{equation}

\paragraph{Case $m=3$:}\label{pencils3}
 
Choose three distinct points on each line,
\begin{equation}\label{points3}
    p_0, p_1, p_2\in\ell_0\,,\quad q_0, q_1, q_2\in\ell_1\,,\quad r_0, r_1, r_2\in\ell_2\,,
\end{equation}
and consider the pencil of cubic curves passing through these points. For this to work, the base points \eqref{points3} of the pencil should be chosen subject to one overall constraint.
Two further degrees of freedom can be eliminated by applying projective transformation preserving the three lines. Hence, we have a six-parameter family of such pencils up to projective equivalence.

More concretely, assuming the lines are chosen as in \eqref{lines}, we have a pencil \eqref{penh} where
\begin{equation}\label{pen3}
Q=y^3+Q_1(x,w)y^2+Q_2(x,w)y+Q_3(x,w)\,,\qquad P(x,w)=w(x-e_1w)(x-e_2w)\,.
\end{equation}
The base points of the pencil are found by intersecting the cubic $Q=0$ with the lines:
\begin{equation}\label{nu30}
    p_i=(1:-\alpha_i:0)\,,\quad q_i=(e_1:\beta_i:1)\,,\quad r_i=(e_2:\gamma_i:1)\,.
\end{equation}
Writing $Q_1=a_{11}x+a_{12}w$, we find that $\alpha_0+\alpha_1+\alpha_2=a_{11}$, $\beta_0+\beta_1+\beta_2=-a_{11}e_1-a_{12}$, $\gamma_0+\gamma_1+\gamma_2=-a_{11}e_2-a_{12}$. Hence, 
the parameters $\alpha_i, \beta_i, \gamma_i$ satisfy the constraint
\begin{equation}
    \sum_i\alpha_{i}+\sum_i\frac{\beta_{i}}{e_1-e_2}+ \sum_i\frac{\gamma_{i}}{e_2-e_1}=0\,.
\end{equation}
Furthermore, by a transformation $y\mapsto y+ax+bw$ we can make $Q_1=0$ bringing $Q$ to the form
\begin{equation}\label{pen31}
Q=y^3+Q_2(x,w)y+Q_3(x,w)\,.
\end{equation}
In that case,
\begin{equation}\label{nu31}
\alpha_0+\alpha_1+\alpha_2=\beta_0+\beta_1+\beta_2=\gamma_0+\gamma_1+\gamma_2=0\,.
\end{equation}

The Seiberg--Witten differential \eqref{sw1} in homogeneous coordinates becomes
\begin{equation}\label{sw2}
 \lambda=\frac{y\,(wdx-xdw)}{w(x-e_1w)(x-e_2w)}\,.
\end{equation}
Its residues at $w=0$ and $x=e_{1,2}w$ are $\alpha_{1,2,3}$, $(e_1-e_2)^{-1}\beta_{1,2,3}$, and $(e_2-e_1)^{-1}\gamma_{1,2,3}$, respectively. 

Thus, the geometric parameters of the pencil are directly related to the residues of $\lambda$ (linear masses).  Generically, we have $3$ distinct residues for each of $x=\infty, e_1, e_2$; we express this by saying that the {\it pattern of residues} of $\lambda$ is $(111), (111), (111)$.

\paragraph{Case $m=4$:}\label{pencils4}
In this case we need a pencil of curves of degree 4 with two double points. We choose ten distinct points on the lines $\ell_0, \ell_1, \ell_2$ as follows:
\begin{equation}\label{points4}
    p_0, p_1, p_2, p_3\in\ell_0\,,\quad q_0, q_1, q_2, q_3\in\ell_1\,,\quad r_0, r_1\in\ell_2\,.
\end{equation}
The curves of the pencil are quartic curves passing through 
\begin{equation}
    (p_0p_1p_2p_3q_0q_1q_2q_3r_0^2r_1^2)\,.
\end{equation}
This notation means that each curve of the pencil should have an ordinary double point at both $r_0$ and $r_1$. By the same reasoning as above, there is a seven-parameter family of such pencils up to projective equivalence. The generic curves in the pencil are quartics with two double points, of geometric genus one.  

Assuming that the lines are of the form \eqref{lines}, we consider a pencil \eqref{penh}, with $Q$ homogeneous of degree $4$ and with
\begin{equation}\label{pen4}
P(x,w)=w(x-e_1w)(x-e_2w)^2\,.
\end{equation}
The quartic $Q=0$ intersects the lines at points
\begin{equation}\label{nu40}
    p_i=(1:-\alpha_i:0)\,,\quad q_i=(e_1:\beta_i:1)\,,\quad r_i=(e_2:\gamma_i:1)\,.
\end{equation}
As before, we find that the parameters describing the points are constrained by
\begin{equation}
    \sum_i\alpha_{i}+\sum_i\frac{\beta_{i}}{e_1-e_2}+ {2} \sum_{i=0,1}\frac{\gamma_{i}}{e_2-e_1}=0\,.
\end{equation}
Furthermore, by a linear transformation $y\mapsto y+ax+bw$ we make $Q_1=0$ bringing $Q$ to the form
\begin{equation}\label{pen41}
Q=y^4+Q_2(x,w)y^2+Q_3(x,w)y+Q_4(x,w)\,.
\end{equation}
In that case, we have 
\begin{equation}\label{nu41}
\alpha_0+\alpha_1+\alpha_2+\alpha_3=\beta_0+\beta_1+\beta_2+\beta_3=\gamma_0+\gamma_1=0\,.  
\end{equation}
We normalise the curves by making each double point $r_0, r_1$ into a pair of distinct points. The Seiberg--Witten differential \eqref{sw2} has simple poles at $4$ points over each of $x=\infty, e_{1}, e_2$, with residues equal to $\alpha_{0,1,2,3}$, $(e_1-e_2)^{-1}\beta_{0,1,2,3}$, and $(e_2-e_1)^{-1}\gamma_{0,1}$ (twice).
We see that the residues of the SW differential are directly related to the geometric parameters of the pencil, and the pattern of residues is $(1111), (1111), (22)$.

\paragraph{Case $m=6$:}\label{pencils6}

This time we need a pencil of curves of degree six. Choose eleven points as follows:
\begin{equation}\label{points6}
    p_0, p_1, p_2, p_3, p_4, p_5\in\ell_0\,,\quad q_0, q_1, q_2, \in\ell_1\,,\quad r_0, r_1\in\ell_2\,.
\end{equation}
The curves of the pencil are of degree six, required to pass through 
\begin{equation}
    (p_0p_1p_2p_3p_4p_5q_0^2q_1^2q_2^2r_0^3r_1^3)\,.
\end{equation}
This notation means that each curve of the pencil should have an ordinary double point at each of $q_{0,1,2}$ and an ordinary triple point at $r_0$ and $r_1$. By the same reasoning, there is an eight-parameter family of such pencils up to projective equivalence. The generic curves in the pencil are sextics with three double and two triple points, of geometric genus one.  

Assuming that the lines $\ell_0, \ell_1, \ell_2$ are brought to the form \eqref{lines}, we obtain a pencil of the form \eqref{penh}, with $Q$ homogeneous of degree $6$ and with
\begin{equation}\label{pen42}
P(x,w)=w(x-e_1w)^2(x-e_2w)^3\,.
\end{equation}
The sextic $Q=0$ intersects the lines at points
\begin{equation}\label{nu60}
    p_i=(1:-\alpha_i:0)\,,\quad q_i=(e_1:\beta_i:1)\,,\quad r_i=(e_2:\gamma_i:1)\,.
\end{equation}
The eleven parameters $\alpha_{0,1,2,3,4,5}$, $\beta_{0,1,2}$, $\gamma_{0,1}$ are constrained by 
\begin{equation}
    \sum_i\alpha_{i}+2\sum_i\frac{\beta_{i}}{e_1-e_2}+ {3} \sum_i\frac{\gamma_{i}}{e_2-e_1}=0\,.
\end{equation}
We normalise the curves by making each double point $q_{0,1,2}$ into a pair of distinct points, and each triple point $r_{0,1}$ into three distinct points. The Seiberg--Witten differential \eqref{sw2} has simple poles at the six points over each of $x=\infty, e_{1}, e_2$, with the residues equal to $\alpha_i$, $(e_1-e_2)^{-1}\beta_i$ (repeated twice), and $(e_2-e_1)^{-1}\gamma_i$ (thrice).  The pattern of residues is therefore $(111111), (222), (33)$.

When $Q$ is brought into the form
\begin{equation}\label{pen61}
Q=y^6+Q_2(x,w)y^4+Q_3(x,w)y^3+Q_4(x,w)y^2+Q_5(x,w)y+Q_6(x,w)\,,
\end{equation}
then   
\begin{equation}\label{nu61}
\alpha_0+\alpha_1+\alpha_2+\alpha_3+\alpha_4+\alpha_5=\beta_0+\beta_1+\beta_2=\gamma_0+\gamma_1=0\,.   
\end{equation}
We see, once again, that the geometric parameters of the pencil are directly related to the linear masses. 

\paragraph{Case $m=2$:}\label{pencils2} For completeness, let us give the result for $m=2$, although in this case the answer is well known. This time, we work in the weighted projective plane $\mathbb{P}^2_{1,2,1}$, so $\deg x=\deg w=1$, $\deg y=2$. It has a singular point $(0:1:0)$, and we choose four lines $\ell_{0,1,2,3}$ (which automatically pass through that singular point)\footnote{Another option is to  work in $\mathbb P^2$, in which case the elliptic pencil will have $9$ base points, $3$ of which infinitesimally close (see, for instance, \cite{KMNOY}, Section~2).}. By a linear transformation of $x,w$ we can bring the lines to
\begin{equation}\label{lines2}
    \ell_0:\  w=0\,,\quad\ell_1:\ x-e_1w=0\,,\qquad\ell_2:\ x-e_2w=0\,, \qquad\ell_3:\ x-e_3w=0\,.
\end{equation}
Choose eight distinct points 
\begin{equation}\label{points2}
    p_0, p_1\in\ell_0\,,\quad q_0, q_1\in\ell_1\,,\quad r_0, r_1\in\ell_2\,,\quad s_0, s_1\in\ell_3\,,
\end{equation}
and consider a pencil of curves of weighted homogeneous degree four passing through these points. Hence, the pencil is of the form \eqref{penh}, with  
\begin{equation}
    Q=y^2+Q_1(x,w)y+Q_2(x,w)\,,\qquad \deg Q_1=2\,,\quad \deg Q_2=4\,,
\end{equation}
and        
\begin{equation}
    \qquad P(x,w)=w(x-e_1w)(x-e_2w)(x-e_3w)\,.
\end{equation}
The generic curves in the pencil are smooth elliptic curves. Write
\begin{equation}\label{nu20}
    p_i=(1:-\alpha_i:0)\,,\quad q_i=(e_1:\beta_i:1)\,,\quad r_i=(e_2:\gamma_i:1)\,,\quad s_i=(e_3:\delta_i:1)\,.
\end{equation}
Then the condition that the curve $Q=0$ passes through these points implies that 
\begin{equation}
  \alpha_{0}+\alpha_1+\frac{\beta_{0}+\beta_1}{(e_1-e_2)(e_1-e_3)}+\frac{\gamma_{0}+\gamma_{1}}{(e_2-e_1)(e_2-e_3)}+\frac{\delta_{0}+\delta_1}{(e_3-e_1)(e_3-e_2)}=0\,.  
\end{equation}
By a change of coordinates $y\mapsto y+ax^2+bxw+cw^2$ we can make $Q_1=0$, in which case 
\begin{equation}\label{nu21}
\alpha_0+\alpha_1=\beta_0+\beta_1=\gamma_0+\gamma_1=\delta_0+\delta_1=0\,.
\end{equation}
This reduces the set of parameters to $\alpha_0, \beta_0, \gamma_0, \delta_0$, plus the modular parameter, the cross-ratio of $\infty, e_{1,2,3}$.  

The Seiberg--Witten differential \eqref{sw1} in homogeneous coordinates is
\begin{equation}\label{sw22}
 \lambda=\frac{y\,(wdx-xdw)}{w(x-e_1w)(x-e_2w)(x-e_3w)}\,.
\end{equation}
Its residues at the points $p_{0,1}$, $q_{0,1}$, $r_{0,1}$, $s_{0,1}$ are equal to 
\begin{equation}
\alpha_{0,1}\,,\quad \frac{\beta_{0,1}}{(e_1-e_2)(e_1-e_3)}\,,\quad\frac{\gamma_{0,1}}{(e_2-e_1)(e_2-e_3)}\,,\quad \frac{\delta_{0,1}}{(e_3-e_1)(e_3-e_2)}\,.  
\end{equation}
This relates the geometric parameters of the pencil to the linear masses. The pattern of residues is $(11), (11), (11), (11)$.
Finally, an explicit formula for $Q$ is
\begin{equation}\label{exq2}
    Q=y^2+a_0x(x-e_1w)(x-e_2w)(x-e_3w)+\sum_{i=1,2,3}a_iw^2\prod_{j\ne i}^3\frac{(x-e_jw)}{(e_i-e_j)}\,,
\end{equation}
with $a_0=\alpha_0\alpha_1/2$, $a_1=\beta_0\beta_1/2$, $a_2=\gamma_0\gamma_1/2$, $a_3=\delta_0\delta_1/2$.

\begin{remark}
    Some of the above elliptic pencils appear in \cite{PSZ, PSWZ, SST} in the context of discrete integrable maps and non-standard Kahan discretisation. Note also that in very special cases such pencils and the corresponding continuous hamiltonian dynamics appeared in the study of symmetric monopoles \cite{HMM95}.  
\end{remark}

\section{Quantum curves}\label{sec5}

Since the classical hamiltonian $h(p,q)$ has a natural quantum analogue $\hh$ \eqref{hamex}, we obtain a natural quantisation of the fibration $h(p,q)=z$ in the form of a one-parameter family of ordinary differential equations (ODEs) 
\begin{equation}\label{qc}
    \hh\left(q,\hbar\frac{d}{dq}\right)\psi(q, z)=z\psi(q, z)\,,\quad z\in\C\,.
\end{equation}
Because of the duality, we can also replace $q$ and $d/dq$ by $\alpha$ and $d/d\alpha$, and the couplings $c$ by $c^\vee$, to obtain a quantisation of the spectral curve $h^\vee(k,\alpha)=z$ in the form of a differential equation in the $\alpha$-variable:
\begin{equation}\label{qcd}
    \hh^\vee\left(\alpha,\hbar\frac{d}{d\alpha}\right)\psi^\vee(\alpha, z)=z\psi^\vee(\alpha, z)\,,\quad z\in\C\,.
\end{equation}
We will refer to both families of ODEs as {\it quantum curves}, as they are the same up to a change of notation. Moreover, one can put them together to form a linear PDE (cf. \cite{Jeong:2018qpc, jeong2023parallel}) 
\begin{equation}\label{spcq}
  \left[\hh^\vee\left(\alpha,\hbar\frac{\partial}{\partial\alpha}\right)- \hh\left(q,\hbar\frac{\partial}{\partial q}\right) \right]\Psi=0\,
\end{equation}
which admits a one-parameter family of solutions in the form $\Psi=\psi(q, z)\psi^\vee(\alpha, z)$, $z\in\C$. Note a formal similarity between \eqref{spcq} and \eqref{spc}.

Now we would like to characterise the arising ODEs \eqref{qc} intrinsically. First, they are $\Z_m$-invariant and have regular singularities at the fixed points $q=x_i$. To find the leading exponents at singular points, we pick a nonzero fixed point $q=x_i$ and work in a local coordinate $X=q-x_i$. Let us apply $\hh$ to $X^n$, using the formula \eqref{hamex}. By picking the most singular terms, it is easy to see that 
\begin{equation}
    \hh(X^\lambda)=c(\lambda)X^{\lambda-m}+\dots\,,\qquad c(\lambda)=\prod_{j=0}^{m-1}((\lambda-j)\hbar -\mu_j(x_i))\,,
\end{equation}
where the dots denote terms of higher degree in $X$. (The form of $c(\lambda)$ is dictated entirely by the first term, $w_m$, in \eqref{hamex}.) This tells us that the indicial equation determining the local exponents at $q=x_i$ is $c(\lambda)=0$, from which the local exponents are found as
\begin{equation}\label{locex}
\lambda=j+\mu_j(x_i)\hbar^{-1}\,,\qquad j=0,\dots, m-1.    
\end{equation}
The same result is true for $x_i=0$ because that is how we chose the parameters $\alpha_i$. cf. \eqref{ai}.  

Note that in the cases $m=4,6$ we have repetitions among $\mu_j(x_i)$; as a result, some of the local exponents differ by an integer. This is known as {\it resonance}, and in general it leads to the appearance of Jordan blocks in the local monodromy (and logarithmic terms in the local solutions). However, this does not happen in our case. Indeed, by \cite{EM08edo}, Section 6.2, the monodromy representation factors through the orbifold Hecke algebra which is semisimple in our situation (cf. the proof of Theorem 7.1 in \cite{EFMV11ecm}). Hence, we obtain the following result.

\begin{prop}\label{semi} For generic parameters of the Cherednik algebra,
the differential equations \eqref{qc} have local exponents given by \eqref{locex} and \emph{semisimple} local monodromy around each singular point.   \end{prop}
Using the $\Z_m$-symmetry, we may view \eqref{qc} as an equation on the Riemann sphere with the coordinate $x=u(q)$,
\begin{equation}\label{qcr}
    w^{-m}(\hh-z)\psi=0\,,
\end{equation}
where the expressions for $w^{-m}\hh$ and $w^m$ are given in Section \ref{rf}.
We refer to this as \emph{quantum curves in the rational form}. Each is a Fuchsian equation whose singular points are the orbifold points $e_i\in\mathbb{P}^1=\E/\Z_m$. If $\Z_{m_i}$ is the local symmetry group around $e_i$ then the change of coordinate $q\mapsto x$ locally looks like $X\mapsto X^{m_i}$. Therefore, the local exponents of \eqref{qcr} at $x=e_i$ are obtained from \eqref{locex} by dividing by $m_i$. Hence, we obtain the following characterisation of the Fuchsian ODEs \eqref{qcr}.

\begin{prop}\label{semif} For generic parameters $c_l(x_i)$ of the Cherednik algebra, the Fuchsian differential equations \eqref{qcr} have {\it semisimple} local monodromy around each orbifold point $e_i$, with local exponents 
\begin{equation}\label{locexr}
\lambda_j(e_i)=\frac{j+\mu_j(x_i)\hbar^{-1}}{m_i}\,,\qquad j=0,\dots, m-1,    
\end{equation}
where $x_i\in\E$ is any of the preimages of $e_i\in\mathbb P^1$ under the map $q\mapsto x=u(q)$, and $m_i$ is the order of the orbifold point $e_i$.\end{prop}
Below we describe the resulting Fuchsian equations case by case. For simplicity, set $\hbar=1$ so the equations we consider are of the form
\begin{equation}
    \frac{d^m\psi}{dx^m}+A_{1}\frac{d^{m-1}\psi}{dx^{m-1}}+\dots +A_m\psi=0\,.
\end{equation}
Here $A_i=A_i(x)$ are rational functions satisfying $A_i\to 0$ as $x\to\infty$.

\paragraph{Case $m=3$:}
Consider Fuchsian equations of order $3$ with three singular points $x=\infty, e_1, e_2$ (which may be taken as $\infty, 0, 1$) and with prescribed local exponents  
\begin{equation}
    \alpha_{0,1,2}\,,\quad \beta_{0,1,2}, \quad \gamma_{0,1,2}\,,\quad\text{with\ \ }\sum_i\alpha_{i}+\sum_i\beta_{i}+\sum_i\gamma_{i}=3\,. 
\end{equation}
This means that the local monodromy around $x=\infty$ has eigenvalues $e^{2\pi i\alpha_{0,1,2}}$, and similarly for $x=e_{1,2}$. The condition on the sum of local exponents is known as the {\it Fuchs relation}. 
Typically, prescribing local exponents does not determine the equation uniquely: there may  be additional parameters called {\it accessory parameters}. In this case there is one such parameter, corresponding to the variable $z$ in \eqref{qcr}.

\paragraph{Case $m=4$:}
Consider Fuchsian equations of order $4$ with three singular points $x=\infty, e_1, e_2$ and with prescribed local exponents  
\begin{equation}
    \alpha_{0,1,2,3}\,,\quad \beta_{0,1,2,3}, \quad \gamma_{0,1},\quad 1+\gamma_{0,1}\,,\quad\text{with\ \ }
    \sum_i\alpha_{i}+\sum_i\beta_{i}+ {2} \sum_i\gamma_{i}=4\,. 
\end{equation}
The last condition says that the total sum of local exponents is $6$, which is the Fuchs relation. 
In addition, in this case, some of the local exponents differ by an integer, so one expects local solutions to contain logarithms. Imposing {\it semi-simplicity} of the local monodromy (absence of logarithms) at $x=e_2$, one obtains a one-parametric family of such equations, with one accessory parameter, $z$. 

\paragraph{Case $m=6$:}
Consider Fuchsian equations of order $6$ with three singular points $x=\infty, e_1, e_2$ and with prescribed local exponents  
\begin{equation}
    \alpha_{0,1,2,3,4,5}\,,\quad \beta_{0,1,2},\quad 1+\beta_{0,1,2}\,, \quad \gamma_{0,1},\quad 1+\gamma_{0,1},\quad 2+\gamma_{0,1}\,,\quad\sum_i\alpha_{i}+ {2} \sum_i\beta_{i}+ {3} \sum_i\gamma_{i}=6\,. 
\end{equation}
Again, since some local exponents differ by an integer one expects local solutions to contain logarithmic terms. Imposing the condition of semi-simplicity of the local monodromy at $x=e_{1,2}$, one obtains a one-parameter family of such equations, with one accessory parameter, $z$.

\paragraph{Case $m=2$:} This is the well-known case of the {\it Heun equation}, a Fuchsian equation of order $2$ with four singular points $x=\infty, e_1, e_2, e_3$ (which may be taken as $\infty, 0, 1, t$) and with prescribed local exponents  
\begin{equation}
    \alpha_{0,1}\,,\quad \beta_{0,1},\quad \gamma_{0,1}\,, \quad \delta_{0,1}\,,\quad\text{with\ \ }\alpha_{0}+\alpha_{1}+\beta_{0}+\beta_1+\gamma_0+\gamma_1+\delta_0+\delta_1=2\,. 
\end{equation}
Such Fuchsian equations have one accessory parameter, $z$.

\begin{remark}
    Comparing the parameters $\alpha_i$, etc., in the above description with the local exponents given in Proposition \ref{semif}, we see that they are normalised so that 
    \begin{align*}
    \sum_i\alpha_i&=\sum_i\beta_i=\sum_i\gamma_i=1\quad &(m=3)\,,\\
    \sum_i\alpha_i&=\sum_i\beta_i=3/2\,,\quad \sum_i\gamma_i=1/2 &\quad(m=4)\,,\\
    \sum_i\alpha_i&=5/2\,,\quad \sum_i\beta_i=1\,,\quad \sum_i\gamma_i=1/2 &\quad(m=6)\,,\\
    \sum_i\alpha_i&=\sum_i\beta_i=\sum_i\gamma_i=\sum_i\delta_i=1/2\quad &(m=2)\,.
    \end{align*}
Such a normalisation can be always achieved by a suitable gauge transformation. \end{remark}

\section{Further connections}\label{sec6}

Let us describe some other contexts where closely related objects appear.

\subsection{Hitchin systems}
Hitchin systems on algebraic curves are known as a rich source of complex integrable systems \cite{Hitchin87}. For curves of genus $g=0,1$ one needs to allow Higgs fields to have poles \cite{GN94, Markman94, Nekr96}, and several many-body integrable systems have already been identified with Hitchin systems in that way. Our cases can be interpreted as Hitchin systems on the orbifold tori which are Riemann spheres with with three or four punctures as shown in Fig.~\ref{fig:classS}. This agrees with the class-S theory description given in \cite{BBT09}, and, equivalently, M-theory orbifold construction in \cite{Gukov:1998kt}.

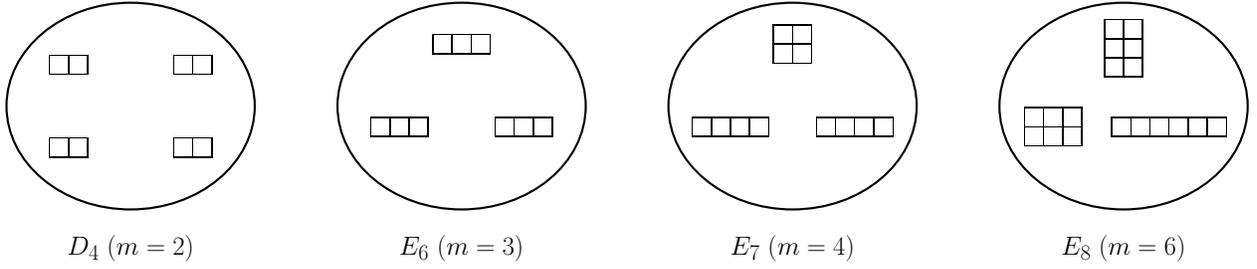
\begin{figure}[h]
\centering
\begin{tikzpicture}[thick,scale=0.55, every node/.style={scale=0.40}]
\begin{scope}[xshift=-16cm]
   \draw (0,0) ellipse (3cm and 2.5 cm);
    \node at (1.5,1) {\ydiagram{2}};
    \node at (1.5,-1) {\ydiagram{2}};
    \node at (-1.5,1) {\ydiagram{2}};
    \node at (-1.5,-1) {\ydiagram{2}};
    \node[anchor=north] at (0,-3){\huge $D_4\  (m=2)$};
\end{scope}
\begin{scope}[xshift=-8cm]
     \draw (0,0) ellipse (3cm and 2.5 cm);
    \node at (0,1.5) {\ydiagram{3}};
    \node at (-1.5,-0.5) {\ydiagram{3}};
    \node at (1.5,-0.5) {\ydiagram{3}};
    \node[anchor=north] at (0,-3){\huge $E_6\  (m=3)$};
\end{scope}
\begin{scope}[xshift=0cm]
        \draw (0,0) ellipse (3cm and 2.5 cm);
    \node at (0,1.5) {\ydiagram{2,2}};
    \node at (-1.5,-0.5) {\ydiagram{4}};
    \node at (1.5,-0.5) {\ydiagram{4}};
    \node[anchor=north] at (0,-3){\huge $E_7\  (m=4)$};
\end{scope}
\begin{scope}[xshift=8cm]
      \draw (0,0) ellipse (3cm and 2.5 cm);
    \node at (0,1.4) {\ydiagram{2,2,2}};
    \node at (-1.7,-0.5) {\ydiagram{3,3}};
    \node at (1.1,-0.5) {\ydiagram{6}};  
    \node[anchor=north] at (0,-3){\huge $E_8\  (m=6)$};
\end{scope}
\end{tikzpicture}
\caption{Three and four punctured spheres corresponding to $(T^2\times\C)/\Z_m$ with fixed points (``punctures'') labeled by Young diagrams representing partitions of $m$.}
\label{fig:classS}
\end{figure}
In the following, we explain how to obtain the Hitchin system description from the Lax presentation.  First, by looking at the properties of the Lax matrix $L=L(\alpha)$, we can recognize $\tilde\phi:=L(\alpha)d\alpha$ as a $\Z_m$-equivariant Higgs field on $\E$. Namely, take $\C\times\C^m$ with the $\Z_m$-action $\omega.(\alpha, \xi)=(\omega^{-1}\alpha, S\xi)$, where $S$ is the matrix \eqref{s}. This induces a $\Z_m$-action on the total space of $\tilde E:=\oplus_{i=0}^{m-1}\mathcal L_{\omega^{-i}q}$, where the line bundle $\mathcal L_q$ was defined in Remark \ref{pic}. This gives a $\Z_m$-equivariant element $\tilde\phi\in\mathrm{End}\tilde E\otimes \Omega^1_\E(\sum x_i)$. At the punctures $\tilde\phi$ has simple poles with $\mathrm{res}\,\tilde\phi|_{\alpha=x_i}$ diagonalisable, and with prescribed eigenvalues $\mu_j^\vee(x_i)$. Further, by modifications at the fixed points one can make the pair $(\tilde E, \tilde\phi)$ $\Z_m$-invariant; the modified pair can then be obtained as a pullback of some Higgs bundle $(E, \phi)$ on $\mathbb P^1=\E/\Z_m$. That is how in general one relates equivariant Higgs bundles with Higgs bundles on orbifolds, which in turn are identified with (weakly) parabolic Higgs bundles \cite{NS95} (see also \cite{groechenig2014hilbert, garcia2020finite} for related studies). The conclusion is that the phase space of our integrable system through the map $(p,q)\mapsto L(\alpha)\mapsto\tilde\phi\mapsto\phi$ gets identified with an open subset of the moduli space $\mathcal M^\vee$ of parabolic $\mathrm{GL}_m$ Higgs bundles on $\mathbb P^1$. The parabolic data consists of the eigenvalues/eigenspaces of the residues of $\phi$ at the orbifolded points (punctures) in $\mathbb P^1$, with three punctures for $m=3,4,6$ and four punctures for $m=2$. (The definition of $\mathcal M^\vee$ also requires a choice of parabolic weights, but this is not important since we work on an open subspace of the moduli space.) The moduli space $\mathcal M^\vee$ carries a structure of a complex  integrable system: the Hitchin fibration and Hitchin system. This identifies our integrable system with the Hitchin system on an open subset of $\mathcal M^\vee$. The Hitchin fibration is built from the family of spectral curves which coincide with our spectral curves $\Sigma_z^\vee=\Ssigma_z^\vee/\Z_m$. Since these have genus one, the fibers are isomorphic to $\Sigma_z^\vee$, $z\in\C$. The dynamics of the Hitchin system is linear along the fibers. Because the dynamics in $p,q$ coordinates along the elliptic curves $\Sigma_z$ is also linear, we conclude that
\begin{equation}\label{ms}
    \Sigma_z\cong \Sigma_z^\vee\quad \forall\ z
\end{equation}
(where $z$ is assumed to be generic so that $\Sigma_z, \Sigma_z^\vee$ are non-singular).
This property is non-obvious; combined with \eqref{dul} it is reminiscent of the SYZ-type {\it mirror symmetry} for Hitchin fibrations due to Donagi--Pantev \cite{DP}. See also Remark \ref{rai} below. 

\begin{remark}
  Here is the sketch of how to observe \eqref{ms} without resorting to Hitchin systems. Starting from the Lax matrix $L(p,q;\alpha)$ and its spectral curve $\Ssigma^\vee:\, \det(L-kI)=0$, we view the family of eigenlines $(L-kI)\ell=0$ parameterised by $(k,\alpha)\in\Ssigma$ as a line bundle $\mathcal L$ over $\Ssigma^\vee$. 
  The dynamics $p(t), q(t)$ induces a dynamics $\mathcal L(t)$ on the Jacobian $\mathrm{Jac}(\Ssigma^\vee)$. One then checks the following two properties: (1) the induced dynamics on $\mathrm{Jac}(\Ssigma^\vee)$ is linear, and (2) $\mathcal L^{s}\cong \mathcal L$ for any $s\in\Z_m$. As a result, the linear motion along $\Sigma$ in the phase space is mapped onto a linear motion along a $\Z_m$-fixed subtorus in the Jacobian of $\Ssigma^\vee$,  
  which is isomorphic to $\mathrm{Jac}(\Ssigma^\vee/\Z_m)\cong \Sigma^\vee$. This implies the isomorphism \eqref{ms}.            
\end{remark}

\begin{remark} In the case $m=4$, the quantum hamiltonian $\hh$ appeared in the studies of multipoint conformal blocks \cite{Buric:2021ttm}.
\end{remark}

\subsection{Local systems, star-shaped quivers, and generalised DAHAs }
\label{localsys}

According to the non-abelian Hodge correspondence \cite{Si1, Si2}, the moduli space of Higgs bundles (the Dolbeaut space $\mathcal M_{Dol}$) over a complex algebraic curve $X$, has two other avatars, $\mathcal M_{dR}$ (de Rham) and $\mathcal M_{B}$ (Betti). The three spaces are diffeomorphic: $\mathcal M_{Dol}$ and $\mathcal M_{dR}$ are obtained from each other by rotating the complex structure within a hyper-K\"ahler family, while $\mathcal M_{dR}$ and $\mathcal M_{B}$ are identified as complex-analytic spaces by the Riemann--Hilbert correspondence. Note that we need a version of $\mathcal M_{Dol}$ for curves with punctures and the so-called weakly parabolic Higgs bunbdles. The existence of a hyper-K\"ahler structure on $\mathcal M_{Dol}$ for this case is due to Biquard and Bolach \cite{wNABH}, cf. \cite{Si2, Kon93, Nak96} for earlier partial results. 

From that perspective, if $\mathcal M$ is one of our moduli spaces of Higgs bundles on the punctured Riemann sphere, then the corresponding de Rham moduli space is precisely one of the four spaces of Fuchsian systems considered by Boalch \cite{boalch09}. As he explains, these moduli spaces are nothing but the ALE spaces considered by Kronheimer \cite{Kr} which can also be recast as quiver varieties \cite{Na} associated with the affine Dynkin quivers of type ${D}_4$, ${E}_6$, ${E}_7$, ${E}_8$ (these are precisely the star-shaped affine Dynkin quivers).  
Note that there are corresponding 3d $\mathcal{N} =4$ quiver gauge theories which are the {\em mirror} theories for the circle reduction of $D_4$, $E_6$, $E_7$, and $E_8$ theories \cite{ISmirror96, BTXmirror}. 
The ${D}_4$ case ($m=2$ case in our language) corresponds to the family of $2\times 2$ Fuchsian systems on $\mathbb P^1$ with four singularities and prescribed local exponents at the singular points; it has been studied from various angles in Painlev\'e theory and related contexts, see in particular \cite{Arinkin, ArinkinLysenko, ArinkinBorodin, Oblezin, LoraySaitoSimpson}. The $E_{6,7,8}$ cases ($m=3,4,6$) are also closely related to Painlev\'e theory, but to difference rather than continuous Painlev\'e equations. As explained in Sections 6 and 7 of \cite{boalch09}, these are essentially the surfaces from Sakai's list \cite{Sakai} within his geometric approach to Painlev\'e equations (they correspond to the cases {\it Add1, Add2, Add3} in \cite{Sakai}). From that perspective, our duality $c\mapsto c^\vee$ and \eqref{ms} appears to be similar to the Okamoto transformation for Painlev\'e VI, interpreted in terms of the middle convolution in \cite{filipuk2006symmetries, boalch09} (cf. Remark \ref{rai} below).

On the Betti side, we have spaces of the monodromy data of the above Fuchsian systems. According to the general theory \cite{CBShaw}, these are modelled by multiplicative quiver varieties associated to the affine Dynkin quivers of type ${D}_4$, ${E}_{6,7,8}$. From yet another perspective, they appear in the work of Etingof, Oblomkov, and Rains \cite{EOR} on generalised double affine Hecke algebras, or DAHAs. These varieties can be characterised as certain affine del Pezzo surfaces, see Sections 6 and 9 of \cite{EOR}.

\begin{remark}\label{rai}
    As Eric Rains pointed out to us, the duality isomorphism \eqref{ms} can be seen from the Weyl-group action on $\mathcal M_{dR}$. Indeed, $\mathcal M_{dR}$ being a quiver variety of type $D_4/E_6/E_7/E_8$ admits such an action, see \cite{crawley2003matrices, boalch09}. One then checks that the transformation $c\mapsto c^\vee$ can be identified with a suitable element of the Weyl group (see Appendix \ref{a:weyl}). By taking a limit to the Higgs moduli space, we conclude that the corresponding Hitchin fibrations are isomorphic.          
\end{remark}

\subsection{Quantum curves and opers} 

As explained above, we can view the fibration $\{\Sigma^\vee_z\}$ on $T^*\E/\Z_m$ as a Hitchin fibration over a punctured $\mathbb P^1$. Quantisation of $\Sigma^\vee_z$ gives a pencil of Fuchsian ODEs with prescribed local monodromy data. The monodromy around each puncture is semi-simple and has prescribed eigenvalues (with repetitions in cases $m=4, 6$). If we write $G=\mathrm{GL}_m$ and denote by $M_i\in G$ the monodromy around $x=e_i$, then $M_i$ belongs to a particular semisimple conjugacy class $[\Lambda_i]:=\{g\Lambda_ig^{-1}\,|\,g\in G\}$. For example, for $m=6$ $\Lambda_{0}$ is a generic diagonal matrix while $\Lambda_{1,2}$ are of the form $\mathrm{diag}(a, a, b, b, c, c)$ and $\mathrm{diag}(d, d, d, e, e, e)$, respectively; the global monodromy in this case represents a point on the character variety
\begin{equation}
    \mathcal M_B:=\{M_0, M_1, M_2\in G\,|\, M_0M_1M_2=\Id\,,\   M_i\in [\Lambda_i]\}\,//\,G\,,
\end{equation}
which is the Betti moduli space mentioned above. (These character varieties are  precisely the affine del Pezzo surfaces from \cite{EOR}.) 
Each quantum curve can also be viewed as a rank $m$ trivial bundle over $\mathbb P^1$ with (flat) connection, so it represents a point in the de Rham space, $\mathcal M_{dR}$, and is in the form of a $\mathrm{GL}_m$-oper. We therefore observe that the pencil of quantum curves can be associated with the one-dimensional Lagrangian subvariety of opers, $\mathcal{L}\subset \mathcal M_{dR}$. This illustrates the general philosophy, going back to Nekrasov--Rosly--Shatashvili \cite{NRS} and Gaiotto \cite{Gaiotto:2014bza}, that quantizing spectral curves of a Hitchin system should produce the variety of opers in the corresponding de Rham moduli space. Note that for compact curves of genus $\ge 2$, a result of that kind has been established in \cite{DumitrescuMulase}, but the case of curves with punctures remains open in general. Note also that in the case of superconformal gauge theories, e.g., $SU(2)$ gauge theory with $N_f =4$, the quantum spectral curves can be studied with the help of instanton counting \cite{Jeong:2018qpc}. Additionally, opers for the massless $E_6$ theory have been investigated through the exact WKB method in \cite{HollandsN_WKBT3}.

\subsection{5d theories} There is an approach to 4d $\mathcal N=2$ SQFTs which allows us to view them as a result of compactifying a 5d theory on a circle. It is then natural to expect that the classical and quantum curves of the 4d theory can be obtained as a limit of the corresponding 5d families. For 5d theories that can be constructed in string theory using five-brane webs, there are systematic approaches for deriving the SW curves on $\mathbb{R}^4 \times S^1$ \cite{Aharony:1997bh, KimYagi14}. The curve can be expressed in terms of a polynomial equation in $ (t, w) \in \mathbb{C}^\ast \times \mathbb{C}^\ast$ with monomials associated with the vertices of a 2d dot diagram which is the dual graph of a 5-brane web, and the coefficients encode the moduli and parameters of the 5d theories.  This is particularly applicable to the 5d theories corresponding to the 4d $D_4$ and $E_{6,7,8}$ theories, which are known as Seiberg's $E_n$ theories. The schematic representation of the webs for Seiberg's theories is shown in Fig.~\ref{fig:5dweb}. 
\begin{figure}[htbp]
\centering
\begin{tikzpicture}[thick,scale=0.46, every node/.style={scale=0.5}]
\begin{scope}[xshift=-9cm]
\draw (1.0,-3) -- (1.0,0);
\draw (-1.0,-3) -- ( -1.0,0);
\draw[fill] (1.0,-3) circle (0.15);
\draw[fill] (-1.0,-3) circle (0.15);
\draw (1.0,3) -- (1.0,0);
\draw (-1.0,3) -- ( -1.0,0);
\draw[fill] (1.0,3) circle (0.15);
\draw[fill] (-1.0,3) circle (0.15);

\draw (-3,1.0) -- (0, 1.0);
\draw (-3,-1.0) -- (0, -1.0);
\draw[fill] (-3,1.0) circle (0.15);
\draw[fill] (-3,-1.0) circle (0.15);

\draw (3,1.0) -- (0, 1.0);
\draw (3,-1.0) -- (0, -1.0);
\draw[fill] (3,1.0) circle (0.15);
\draw[fill] (3,-1.0) circle (0.15);

\fill[fill = white] (0,0) circle (2);
\fill[gray!30] (0,0) circle (2);
    \node[anchor=north] at (0,-4){\LARGE $E_5$};
\end{scope}

\begin{scope}[xshift=0cm]
\draw (1.0,-3) -- (1.0,0);
\draw (0.0,-3) -- (0,0);
\draw (-1.0,-3) -- ( -1.0,0);
\draw[fill] (1.0,-3) circle (0.15);
\draw[fill] (0,-3) circle (0.15);
\draw[fill] (-1.0,-3) circle (0.15);

\draw (-3,-1.0) -- (0, -1.0);
\draw (-3,0.0) -- (0, 0);
\draw (-3,1.0) -- (0, 1.0);
\draw[fill] (-3,-1.0) circle (0.15);
\draw[fill] (-3,0) circle (0.15);
\draw[fill] (-3,1.0) circle (0.15);

\draw (1.8,3.2) -- (-0.7,0.7);
\draw (2.5,2.5) -- (0,0);
\draw (3.2,1.8) -- (0.7,-0.7);
\draw[fill] (1.8,3.2) circle (0.15);
\draw[fill] (2.5,2.5) circle (0.15);
\draw[fill] (3.2,1.8) circle (0.15);

\fill[fill = white] (0,0) circle (2);
\fill[gray!30] (0,0) circle (2);
\node[anchor=north] at (0,-4){\LARGE $E_6$};
\end{scope}

\begin{scope}[xshift=9cm]
\draw (1.2,-3) -- (1.2, 0);
\draw (0.4,-3) -- (0.4, 0);
\draw (-0.4,-3) -- (-0.4, 0);
\draw (-1.2,-3) -- ( -1.2, 0);
\draw[fill] (0.8,-3) ellipse (0.45cm and 0.25 cm);
\draw[fill] (-0.8,-3) ellipse (0.45cm and 0.25 cm);

\draw (-3,-1.2) -- (0, -1.2);
\draw (-3,-0.4) -- (0, -0.4);
\draw (-3,0.4) -- (0, 0.4);
\draw (-3,1.2) -- (0, 1.2);
\draw[fill] (-3,-1.2) circle (0.15);
\draw[fill] (-3,-0.4) circle (0.15);
\draw[fill] (-3,0.4) circle (0.15);
\draw[fill] (-3,1.2) circle (0.15);

\draw (3.4,1.6) -- (0.9,-0.9);
\draw (2.8,2.2) -- (0.3,-0.3);
\draw (2.2,2.8) -- (-0.3,0.3);
\draw (1.6,3.4) -- (-0.9,0.9);
\draw[fill] (3.4,1.6) circle (0.15);
\draw[fill] (2.8,2.2) circle (0.15);
\draw[fill] (2.2,2.8) circle (0.15);
\draw[fill] (1.6,3.4) circle (0.15);

\fill[fill = white] (0,0) circle (2);
\fill[gray!30] (0,0) circle (2);
\node[anchor=north] at (0,-4){\LARGE $E_7$};
\end{scope}

\begin{scope}[xshift=18cm]
\draw (-3,-1.0) -- (-1.732,-1.0);
\draw (-3,-0.6) -- (-1.908,-0.6);
\draw (-3,-0.2) -- (-1.99,-0.2);
\draw (-3,0.2) -- (-1.99,0.2);
\draw (-3,0.6) -- (-1.908,0.6);
\draw (-3,1.0) -- (-1.732,1.0);
\draw[fill] (-3,-0.8) ellipse (0.15cm and 0.25cm);
\draw[fill] (-3,0.0) ellipse (0.15cm and 0.25cm);
\draw[fill] (-3,0.8) ellipse (0.15cm and 0.25cm);

\draw (1.0,-3) -- (1.0,1.732);
\draw (0.6,-3) -- (0.6,1.908);
\draw (0.2,-3) -- (0.2,1.99);
\draw (-0.2,-3) -- (-0.2,1.99);
\draw (-0.6,-3) -- (-0.6,1.908);
\draw (-1.0,-3) -- (-1.0,1.732);
\draw[fill] (-0.6,-3) ellipse (0.45cm and 0.25cm);
\draw[fill] (0.6,-3) ellipse (0.45cm and 0.25cm);

\draw (3.5,1.5) -- (1.0,-1.0);
\draw (3.1,1.9) -- (0.6,-0.6);
\draw (2.7,2.3) -- (0.2,-0.2);
\draw (2.3,2.7) -- (-0.2,0.2);
\draw (1.9,3.1) -- (-0.6,0.6);
\draw (1.5,3.5) -- (-1.0,1.0);

\draw[fill] (3.5,1.5) circle (0.15);
\draw[fill] (3.1,1.9) circle (0.15);
\draw[fill] (2.7,2.3) circle (0.15);
\draw[fill] (2.3,2.7) circle (0.15);
\draw[fill] (1.9,3.1) circle (0.15);
\draw[fill] (1.5,3.5) circle (0.15);

\fill[fill = white] (0,0) circle (2);
\fill[gray!30] (0,0) circle (2);
\node[anchor=north] at (0,-4){\LARGE $E_8$};
\end{scope}

\end{tikzpicture}
\caption{Shown above are the five-brane webs corresponding to 5d Seiberg's theories. For all cases, the internal part of the diagram is represented by a large black circle, while only the external legs are illustrated in detail. Black dots are used to represent seven-branes and lines represent five-branes. }
\label{fig:5dweb}
\end{figure}
The SW curves for Seiberg's $E_n$ theories have been obtained using 5-brane webs in \cite{KimYagi14}. They have been further quantised in \cite{MoriyamaQC}, see also \cite{moriyama2021quantum}. We have checked that our results are consistent with those in \cite{KimYagi14, MoriyamaQC}, see \cite{4d5dpaper}. More recently, the classical curves of these $5d$ theories were studied in \cite{Bershtein2024} in connection with $q$-Painlev\'e equations and cluster integrable systems. In that context, a duality similar to the equation \eqref{spc} was observed. We expect that our duality can be recovered as a limit of this duality for 5d curves.       

\section*{Acknowledgement}
We are grateful to P.~Bolach, C.~Closset, P.~Etingof, P.~Vanhaecke, A.~King, O.~Lechtenfeld, K.~Lee, P.~Longhi, J.~Manschot, M.~Martone, M.~Mazzocco, J.~Minahan, N.~Nekrasov, E.~Rains, T.~Schedler, E.~Sklyanin, Y.~Tachikawa, K.~Takemura and A.~Veselov for stimulating discussions and useful comments.
PCA is supported in part by DOE grant DE-SC1019775. YL is supported by KIAS individual grant PG084801.

\appendix

\section{Elliptic functions and duality}
\label{a:duality}

Here we collect the main properties of the elliptic functions used throughout the paper. Associated to the lattice $\L=2\Z\omega_1+2\Z\omega_2$, we have Weierstrass functions $\sigma, \zeta, \wp$. Recall that $\sigma(x)$ is an odd, entire function with $\sigma'(0)=1$ and with the properties
\begin{align} \label{sigma}
   \sigma(x + \gamma) = (-1)^{m n + m + n} e^{\eta(\gamma)(x + \gamma)}\sigma(x)\,,\qquad \gamma = 2 m \omega_1 + 2 n \omega_2\,,
\end{align}
where $\eta(\gamma)$ was defined in \eqref{eta}. Consider the function
\begin{align}
    \varphi(x,z) = \frac{\sigma(x-z)}{\sigma(x) \sigma(-z)},\qquad \varphi(x,z)=-\varphi(z,x)\,.
\end{align}
It has the following translation properties: 
\begin{align}
    \frac{\varphi(x+\gamma,z)}{\varphi(x,z)}  = e^{-\eta(\gamma) z},\qquad
    \frac{\varphi(x, z +\gamma)}{\varphi(x,z)}  = e^{-\eta(\gamma) x}\,,\qquad \gamma\in\L\,.   
\end{align} 
Next, we have 
\begin{align}\label{vv}
v_{l}( x,z) & =\sum _{\{x_{i}\}} c_{l}( x_{i}) e^{-\eta ( \Omega_{l} x_{i}) z} \varphi ( x-x_{i} ,\Omega_{-l} z)\,,\qquad l\in\Z_m\setminus\{0\}\,,
\end{align}
where $\Omega_l=1-\omega^l$ and the summation is over all fixed points $x_i\in\E$. Note that our convention is to set $c_l(x_i)=0$ whenever $x_i$ is {\it not} fixed by $\omega^l$. Hence, the above summation reduces to $x_i\in (\Omega_l)^{-1}\L/\L$. The following properties are now clear:
\begin{align}
 \mathrm{res}_{x=x_i}v_l(x,z) & = c_l(x_i)e^{-\eta(\Omega_{l}x_i)z}\,,
 \\
 v_l(x+\gamma,z) & =e^{-\eta(\gamma) \Omega_{-l} z}v_l(x,z)\,,\quad\gamma\in\L\,.
\end{align}
These properties characterize $v_l$ uniquely.
On the other hand, as a function of $z$, $v_l(x,z)$ has simple poles at fixed points $z=x_i$, with residues 
\begin{align}\label{cc}
  \mathrm{res}_{z =x_i} v_l(x, z)= - \frac{1}{\Omega_{-l}} \sum_{ \{x_j\} } c_l(x_j) e^{\eta(\Omega_{-l} x_i ) x_j - \eta (\Omega_{l} x_j) x_i} e^{- \eta (\Omega_{-l} x_i) x}\,.  
\end{align}
Also, under translations in the $z$ variable one has 
\begin{align}
 v_l(x, z+\gamma) & =e^{-\eta(\gamma) \Omega_{l} x}v_l(x,z)\,,\quad\gamma\in\L\,.
\end{align}
(This uses that $\eta(\Omega_{-l}\gamma) = \Omega_l \eta(\gamma) $ and the property $\eta(a) b - \eta(b) a \in 2 \pi i \mathbb{Z}$ for $a, b\in \L$.) 
Let us define the dual parameters $c^\vee$ by
\begin{equation}\label{dualc}
    c^\vee_{l}(x_i)= \frac{1}{\Omega_{l}} \sum_{ \{x_j\} } c_{-l}(x_j) e^{\eta(\Omega_{l} x_i ) x_j - \eta (\Omega_{-l} x_j) x_i}\,.
\end{equation}
Then \eqref{cc} becomes
\begin{equation}
 \mathrm{res}_{z=x_i}v_l(x,z)  = - c_{-l}^\vee(x_i)e^{-\eta(\Omega_{-l}x_i)x}\,.   
\end{equation}
We conclude that $v_l$ can be uniquely characterised by its properties  in $z$. By comparing the properties in $x$ and $z$, we obtain
the following {\it duality}:
\begin{equation}\label{dual}
    v_{l, c}(x,z)=-v_{-l, c^\vee}(z,x)\,.
\end{equation}

\subsection{Duality transformation}
Let us now describe explicitly the transformation $c\mapsto c^\vee$ given by \eqref{dualc}. We will use the notation $\omega_{1,2,3}$ and $\eta_{1,2}$ as in \eqref{table0}. We will also use the fact that $\zeta(\omega_1) = \pi/(4 \omega_1)$ for the lemniscatic lattice and $\zeta(\omega_1) = \pi/(2 \sqrt{3} \omega_1)$ for the equianharmonic lattice.

For $m=2$, take $x_0=0$, $x_i=\omega_i$ and denote $g_i=c_1(x_i)$, $i=0,1,2,3$. Then
\begin{align} \label{dualmat2}
\begin{pmatrix} 
    g_0^\vee\\ g_1^\vee \\g_2^\vee \\ g_3^\vee
    \end{pmatrix}
     = \frac{1}{2} \begin{pmatrix}
   1 & 1 & 1 & 1 \\  1 & 1 & -1 & -1 \\ 1 & -1 & 1 & -1  \\  1 & -1 & -1 & 1 
    \end{pmatrix} 
    \begin{pmatrix}
    g_0\\ g_1 \\g_2 \\ g_3
    \end{pmatrix}\,.
\end{align}

For $m=3$, take $x_0=0$, $x_{1,2}=\eta_{1,2}$. We have 6 parameters $c_i(x_{0,1,2})$ with $i=1,2$. Set $\vec{c}_i = \begin{pmatrix}
c_i(x_0) & c_i(x_1) & c_i(x_2)
\end{pmatrix}^T$ (similarly for dual variables). Then
\begin{align} \label{dualmat3}
   \begin{pmatrix}
     \vec{c}^{\, \vee}_1\\\vec{c}^{\,\vee}_2
    \end{pmatrix}  =  \begin{pmatrix}
   0 & A \\  B & 0 
    \end{pmatrix} \begin{pmatrix}
   \vec{c}_1 \\ \vec{c}_2
    \end{pmatrix}
\end{align}
where
\begin{align}
\quad A= \frac{1}{1-\omega} \begin{pmatrix}
       1 & 1 & 1 \\
       1 & \omega & \omega^2\\
       1 &  \omega^2 & \omega
    \end{pmatrix},
   \quad  B= \frac{1}{1-\omega^2} \begin{pmatrix}
       1 & 1 & 1 \\
       1 & \omega^2 & \omega\\
       1 &  \omega & \omega^2
    \end{pmatrix}, \quad \omega=e^{2\pi i/3}.
\end{align}

For $m=4$, take $x_0=0$, $x_1=\omega_3$, $x_2=\omega_1$, $x_3=\omega_2$. We have 7 parameters, $c_{1,2,3}(x_{0,1})$ and $ c_2(x_2) =  c_2(x_3)$.  Set $\vec{c}_i = \begin{pmatrix}
c_i(x_0) & c_i(x_1)
\end{pmatrix}^T $ for $i=1,3$. Then 
\begin{align}
   \begin{pmatrix}
c_2^\vee(x_0) \\ c_2^\vee(x_1) \\ c_2^\vee(x_2)
\end{pmatrix}  = \frac{1}{2} \begin{pmatrix}
1 & 1 & 2 \\
1 & 1 & -2 \\
1 & -1 & 0 
\end{pmatrix}
\begin{pmatrix} \label{dualmat4}
c_2(x_0) \\ c_2(x_1) \\ c_2(x_2)
\end{pmatrix}
\,,\qquad
     \begin{pmatrix}
     \vec{c}^{\, \vee}_1\\\vec{c}^{\,\vee}_3
    \end{pmatrix} =  
    \begin{pmatrix}
   0 & C \\  D & 0 
    \end{pmatrix} \begin{pmatrix}
   \vec{c}_1 \\ \vec{c}_3
    \end{pmatrix}
\end{align}
with
\begin{align}
    C= \frac{1}{1-i} \begin{pmatrix}
       1 & 1  \\
       1 & -1
    \end{pmatrix}\,, \quad D= \frac{1}{1+i} \begin{pmatrix}
       1 & 1 \\
       1 & -1
    \end{pmatrix}\,.
\end{align}

Finally, for $m=6$ we take $x_0=0$, $x_{1,2}=\eta_{1,2}$, $x_3=\omega_1$, $x_4=\omega_2$, $x_5=\omega_3$. We have 8 parameters,
$c_i(x_0)$, $i=1,\dots, 5$, $c_i(x_{1})=c_i(x_2)$, $i=2,4$, and $c_3(x_3)=c_3(x_4)=c_3(x_5)$. Set
$\vec{c}_i = \begin{pmatrix}
c_i(x_0) & c_i(x_1)
\end{pmatrix}^T $, for $i=2,4$ and $ c_2(x_1) =  c_2(x_2)$, $ c_3(x_3) =  c_3(x_4) =c_3(x_5)$, $ c_4(x_1) =  c_4(x_2)$. 
Then
\begin{align} \label{dualmat6}
  \begin{pmatrix}
c_1^\vee(x_0) \\ c_5^\vee(x_0) 
\end{pmatrix}   = \begin{pmatrix}
0 & \frac{1}{1-\epsilon} \\
\frac{1}{1-\epsilon^5} & 0 
\end{pmatrix}
\begin{pmatrix}
c_1(x_0) \\ c_5(x_0) 
\end{pmatrix}\,, \quad   
 \begin{pmatrix}
c_3^\vee(x_0) \\ c_3^\vee(x_3) 
\end{pmatrix} = \frac{1}{2} \begin{pmatrix}
1 & 3 \\
1 & -1
\end{pmatrix}
\begin{pmatrix}
c_3(x_0) \\ c_3(x_3) 
\end{pmatrix}
\end{align}
and 
\begin{align*}
   \begin{pmatrix}
     \vec{c}^{\,\vee}_2\\\vec{c}^{\,\vee}_4
    \end{pmatrix}  =  \begin{pmatrix}
   0 & E \\  F & 0 
    \end{pmatrix} 
    \begin{pmatrix}
   \vec{c}_2 \\ \vec{c}_4
    \end{pmatrix}
\end{align*}
where $\epsilon =e^{\pi i/3}$ and
\begin{align}
    E= \frac{1}{1-\epsilon^2} \begin{pmatrix}
       1 & 2 \\
       1 & \epsilon^2 + \epsilon^4
    \end{pmatrix}\,, \quad F= \frac{1}{1-\epsilon^4} \begin{pmatrix}
       1 & 2\\
       1 & \epsilon^2 + \epsilon^4
    \end{pmatrix}\,.
\end{align}

\subsection{Lie-theoretic interpretation of the duality} \label{a:weyl}
The duality transformation $c\mapsto c^\vee$ admits an interpretation in terms of the Weyl group action on quiver varieties. According to \cite{CB01}, any Fuchsian system gives rise to a representation of the deformed preprojective algebra $\Pi^\lambda$ of a star-shaped quiver $Q$. On the physics side, $\Pi^\lambda$ is defined in terms of the F-term equation for the corresponding mirror quiver theory, and the deformation parameter $\lambda$ can be identified with the Fayet-Iliopoulos parameters \cite{KapustinS_HBImpurity, BourgetGG_FIflow3dNeq4, Kapustin_N2q2Compac3d}.
The number of legs of $Q$ equals the number of singular points of the system, and the leg lengths, deformation parameters $\lambda$, and the dimension vector are related to the local monodromy data at the singular points, see \cite{CB01} for details. 
We are interested in the case of Fuchsian systems whose singularities and local monodromy are of the same type as for the quantum curves considered in Sec.~\ref{sec5}, see Proposition \ref{fuchs}. Hence, they are systems of rank $m=2,3,4,6$ with regular singularities at the orbifold points of $\mathbb P^1=\E/\Z_m$. For each orbifold point $e_r$, let $\Z_{m_r}$ be the isotropy group of $e_r$. The local monodromy is required to be semisimple, with $m_r$ distinct eigenvalues of the same multiplicity $m/m_r$. Up to a twist, we may assume that at each singular point the product of local eigenvalues is $1$. Let $\ell$ denote the number of orbifold points, i.e., $\ell=4$ for $m=2$ and $\ell=3$ for $m=3,4,6$. The corresponding star-shaped quiver $Q$ will have the central node (indexed as $0$) and $\ell$ legs, with nodes along the $r$th leg indexed as $[r,s]$, with $0\le r\le \ell-1$ and $1\le s\le m_r-1$. Hence, ignoring orientation of the edges, $Q$ is an affine Dynkin diagram of type $D_4$, $E_6$, $E_7$, ${E_8}$, see Fig.~\ref{fig:affine}. Let $I$ be the set of the vertices of $Q$; they can be identified with simple roots $\alpha_i$ of the affine root system of type $Q$, so elements of $\mathbb Z^I$ are viewed as linear combinations of simple roots. The (affine) Weyl group $W$ of $Q$ acts on $\mathbb Z^I$; it is generated by simple reflection $s_{\alpha_i}$, $i\in I$. We set $\delta=\sum_{i\in I}n_i\alpha_i$ where the integers $n_i$ are indicated by the labels in Fig.~\ref{fig:affine}, and note that $s_{\alpha_i}(\delta)=\delta$ for all $i\in I$.  

\begin{figure}[h]
\centering
\begin{tikzpicture}[scale=0.66, every node/.style={scale=0.8}]

\begin{scope}[xshift= 0 cm]
\draw (180: 1.2)--(0: 1.2 );
\draw (0: 0)--(60: 1.2 );
\draw (0: 0)--(120: 1.2 );
\draw[fill=white] (0:0) circle (5pt) node[anchor=north,yshift=-0.1cm](0){2};
\draw[fill=white] (0: 1.2 ) circle (5pt) node[anchor=north,yshift=-0.1cm](1){1};
\draw[fill=white] (60: 1.2 ) circle (5pt) node[anchor=south,yshift=0.1cm](2){1};
\draw[fill=white] (120: 1.2 ) circle (5pt) node[anchor=south,yshift=0.1cm](3){1};
\draw[fill=white] (180: 1.2 ) circle (5pt) node[anchor=north,yshift=-0.1cm](4){1};
    \node[anchor=north] at (0,-1){\large ${D}_{4}$};
\end{scope}

\begin{scope}[xshift=4.5 cm]
\draw (180: 2.4)--(0: 2.4);
\draw (0: 0)--(90: 2.4 );
\draw[fill=white] (0:0) circle (5pt) node[anchor=north,yshift=-0.1cm](0){3};
\foreach \j in {1,2}
{\pgfmathtruncatemacro{\y}{(3 -\j)};
\draw[fill=white] (0: 1.2*\j ) circle (5pt) node[anchor = north,yshift=-0.1cm]{\y};
}
\foreach \j in {1,2}
{\pgfmathtruncatemacro{\y}{(3 -\j)};
\draw[fill=white] (180: 1.2*\j ) circle (5pt) node[anchor = north,yshift=-0.1cm]{\y};
}
\foreach \j in {1,2}
{\pgfmathtruncatemacro{\y}{(3 -\j)};
\draw[fill=white] (90: 1.2*\j ) circle (5pt) node[anchor = east,xshift=-0.1cm]{\y};
}
    \node[anchor=north] at (0,-1){\large ${E}_{6}$};
\end{scope}

\begin{scope}[xshift=11.5 cm]
\draw (180: 3.6)--(0: 3.6);
\draw (0: 0)--(90: 1.2);
\draw[fill=white] (0:0) circle (5pt) node[anchor=north,yshift=-0.1cm](0){4};
\foreach \j in {1,2,3}
{\pgfmathtruncatemacro{\y}{(4 -\j)};
\draw[fill=white] (0: 1.2*\j ) circle (5pt) node[anchor = north,yshift=-0.1cm]{\y};
}
\foreach \j in {1,2,3}
{\pgfmathtruncatemacro{\y}{(4 -\j)};
\draw[fill=white] (180: 1.2*\j ) circle (5pt) node[anchor = north,yshift=-0.1cm]{\y};
}
\draw[fill=white] (90: 1.2 ) circle (5pt) node[anchor = east,xshift=-0.1cm]{2};
\node[anchor=north] at (0,-1){\large ${E}_{7}$};
\end{scope}

\begin{scope}[xshift=18.5 cm]
\draw (180: 2.4)--(0: 6);
\draw (0: 0)--(90: 1.2);
\draw[fill=white] (0:0) circle (5pt) node[anchor=north,yshift=-0.1cm](0){6};
\foreach \j in {1,...,5}
{\pgfmathtruncatemacro{\y}{(6 -\j)};
\draw[fill=white] (0: 1.2*\j ) circle (5pt) node[anchor = north,yshift=-0.1cm]{\y};
}
\foreach \j in {1,2}
{\pgfmathtruncatemacro{\y}{2*(3 -\j)};
\draw[fill=white] (180: 1.2*\j ) circle (5pt) node[anchor = north,yshift=-0.1cm]{\y};
}
\draw[fill=white] (90: 1.2 ) circle (5pt) node[anchor = east,xshift=-0.1cm]{3};
    \node[anchor=north] at (1,-1){\large ${E}_{8}$};
\end{scope}

\end{tikzpicture}
\caption{Affine Dynkin diagrams of type ${D}_4$ and ${E}_{6,7,8}$.}
\label{fig:affine}
\end{figure}
Introduce parameters $\xi_{r,s}$ as follows:
\begin{equation}\label{xi}
 \xi_{r,s}=\mu_s(x_r)/m_r\,,\qquad 0\le r\le \ell-1\,,\quad 0\le s\le m_r-1\,,   
\end{equation}
where $x_r\in\E$ is any of the preimages of the orbifold point $e_r$, and
$\mu_j(x_i)$ are defined in \eqref{mui}. The $\xi_{r,s}$ encode the local exponents of the Fuchsian systems in question. We now assign parameters to the nodes of $Q$ in accordance with \cite{CB01}:
\begin{equation}\label{lambda}
    \lambda_0=-\sum_{r=0}^{\ell-1}\xi_{r,0}\,,\quad \lambda_{[r,s]}=\xi_{r, s-1}-\xi_{r, s}\,,\qquad 0\le r\le \ell-1\,,\  1\le s\le m_r-1\,,
\end{equation}
and it is not hard to check that
\begin{equation}\label{fuchs}
    \delta\cdot\lambda:=\sum_{i\in I}n_i\lambda_i=0\,.
\end{equation}
It is convenient to think of $\lambda$ as an element of $\left(\C^I\right)^*$, with $\lambda=\sum_{i\in I}\lambda_i\epsilon_i$ where the basis $\{\epsilon_i\}$ is dual to $\{\alpha_i\}$ (thus, $\epsilon_i$'s are the fundamental coweights). 
This allows us to define the dual action of the Weyl group of $Q$ on $\lambda\in\left(\C^I\right)^*$. Because of the relation \eqref{fuchs}, this action reduces to the action of the Weyl group $W_0$ for the \emph{finite} Dynkin diagram $Q_0\subset Q$ of type $D_4, E_6, E_7, E_8$, respectively. 

This allows us to identify the duality transformation $c\mapsto c^\vee$ with a suitable element of $W_0$. Namely, by combining \eqref{mui}, \eqref{xi}, \eqref{lambda}, 
we may consider the change of coordinates $c\mapsto\lambda$. If $i_0\in I$ is the extended node, we can eliminate $\lambda_{i_0}$ using \eqref{fuchs} and work with $(\lambda_i)_{i\ne i_0}$, referring to it as $\lambda$ by a slight abuse of notation. A direct calculation gives the following result.

\begin{prop} The duality transformation $c\mapsto c^\vee$ in terms of $\lambda$ is given by $\lambda\mapsto \lambda^\vee=\lambda B$, where $\lambda$ and $B$ are as follows.

For $m=2$, $\lambda=(\lambda_0,\lambda_{[11]}, \lambda_{[21]},\lambda_{[31]})$ and
\begin{align}\label{b2}
B=\begin{pmatrix}
 2 & -1 & -1 & -1 \\
 1 & 0 & -1 & -1 \\
 1 & -1 & 0 & -1 \\
 1 & -1 & -1 & 0 
\end{pmatrix}
\end{align}

For $m=3$, $\lambda=(\lambda_{[01]},\lambda_0,\lambda_{[11]},\lambda_{[12]}, \lambda_{[21]},\lambda_{[22]})$ and
\begin{align}\label{b3}
B=\begin{pmatrix}
 0 & -1 & 0 & 1 & 0 & 1 \\
 -1 & 0 & -1 & 2 & -1 & 2 \\
 0 & 0 & -1 & 1 & -1 & 2 \\
 0 & 0 & -1 & 1 & 0 & 1 \\
 0 & 0 & -1 & 2 & -1 & 1 \\
 0 & 0 & 0 & 1 & -1 & 1 
\end{pmatrix}
\end{align}

For $m=4$, $\lambda=(\lambda_{[01]},\lambda_{[02]},\lambda_0,\lambda_{[11]},\lambda_{[12]}, \lambda_{[13]},\lambda_{[21]})$ and
\begin{align}\label{b4}
B = \begin{pmatrix}
 0 & 2 & -1 & 0 & 2 & -1 & -1 \\
 0 & 1 & 0 & 0 & 1 & -1 & -1 \\
 -1 & 3 & 0 & -1 & 3 & -1 & -2 \\
 0 & 2 & 0 & -1 & 2 & 0 & -2 \\
 0 & 1 & 0 & -1 & 2 & 0 & -1 \\
 0 & 0 & 0 & 0 & 1 & 0 & -1 \\
 0 & 1 & 0 & -1 & 2 & -1 & -1 
\end{pmatrix}
\end{align}

For $m=6$, $\lambda=(\lambda_{[01]},\lambda_{[02]},\lambda_{[03]}, \lambda_{[04]},\lambda_0,\lambda_{[11]},\lambda_{[12]},\lambda_{[21]})$ and
\begin{align}\label{b6}
B = \begin{pmatrix}
 0 & 4 & -1 & -1 & -1 & -1 & 3 & -2 \\
 0 & 3 & -1 & -1 & 0 & -1 & 2 & -2 \\
 0 & 2 & -1 & -1 & 0 & -1 & 2 & -1 \\
 0 & 1 & -1 & 0 & 0 & 0 & 1 & -1 \\
 -1 & 5 & -1 & -1 & 0 & -2 & 4 & -3 \\
 0 & 3 & -1 & 0 & 0 & -2 & 3 & -2 \\
 0 & 1 & 0 & 0 & 0 & -1 & 2 & -1 \\
 0 & 2 & 0 & -1 & 0 & -1 & 2 & -2
\end{pmatrix}
\end{align}
\end{prop}

Observe that the columns of $B$ have components of the same sign (or zero). This is because the matrix $B$ describes a certain transformation $w_0\in W_0$ written in the basis of simple roots. (The map $\lambda\mapsto \lambda^\vee=\lambda B$ is the dual transformation, i.e., it describes $w_0$ in the dual basis of fundamental coweights.) As it turns out, $w_0$ is a product of commuting reflections. To write it down, we will use the labelling of simple roots $\alpha_i$ as in the tables in \cite{Bourbaki}. The result is as follows. 

For $m=2$, the matrix \eqref{b2} describes the action of the following element $w_0$ in the basis of simple roots $\{\alpha_2, \alpha_1, \alpha_3, \alpha_4\}$:  
\begin{align*}
 w_0=s_\beta\quad\text{where}\ \beta=\alpha_1+\alpha_2+\alpha_3+\alpha_4.   
\end{align*}

For $m=3$, the matrix \eqref{b3} describes the action of the following element $w_0$ in the basis of simple roots $\{\alpha_2, \alpha_4, \alpha_3, \alpha_1, \alpha_5, \alpha_6\}$:
 \begin{align*}
    &  w_0=s_{\beta_1}s_{\beta_2}s_{\beta_3}\quad\text{where $\beta_{1,2,3}$ are mutually orthogonal roots},\\
   & \beta_1 =  \alpha _1+\alpha _2+2 \alpha _3+2 \alpha _4+\alpha _5, \\ 
   & \beta_2 = \alpha _2+\alpha _3+2 \alpha _4+2 \alpha _5+\alpha _6, \\ & \beta_3 = \alpha _2+\alpha _4\,.
\end{align*}

For $m=4$, the matrix \eqref{b4} describes the action of the following element $w_0$ in the basis of simple roots $\{\alpha_3, \alpha_1, \alpha_4, \alpha_5, \alpha_6, \alpha_7, \alpha_2\}$:
\begin{align*}
 &  w_0=s_{\beta_1}s_{\beta_2}s_{\beta_3}\quad\text{where $\beta_{1,2,3}$ are mutually orthogonal roots},\\
   & \beta_1 = \alpha _1+2 \alpha _2+2 \alpha _3+3 \alpha _4+2 \alpha _5+\alpha _6+\alpha _7, \\ 
   & \beta_2 = \alpha _2+\alpha _3+2 \alpha _4+2 \alpha _5+\alpha _6, \\ & \beta_3 = \alpha _3+\alpha _4\,. 
\end{align*}    

For $m=6$, the matrix \eqref{b6} describes the action of the following element $w_0$ in the basis of simple roots $\{\alpha_5, \alpha_6, \alpha_7, \alpha_8, \alpha_4, \alpha_3, \alpha_1, \alpha_2\}$:
\begin{align}
&  w_0=s_{\beta_1}s_{\beta_2}s_{\beta_3}s_{\beta_4}\quad\text{where $\beta_{1,2,3,4}$ are mutually orthogonal roots},\\
   & \beta_1  = \alpha_1 + 2 \alpha_2 +3 \alpha_3 +4 \alpha_4 +3 \alpha_5 +2 \alpha_6 + 2 \alpha_7+\alpha_8, \\ 
   & \beta_2  = \alpha _1+2 \alpha _2+2 \alpha _3+3 \alpha _4+2 \alpha _5+\alpha _6, \\ 
   & \beta_3 = \alpha _2+\alpha _4+\alpha _5+\alpha _6+\alpha _7+\alpha _8, \\
   & \beta_4 =  \alpha _4+\alpha _5\,.
\end{align}

\begin{remark}
The above action of $w_0$ is on the space of parameters of $\mathcal M_{dR}$. There are also related reflection functors \cite{CB01} mapping between the moduli spaces with different parameters. Their action is known to respect the symplectic structure \cite{Na1} (cf. \cite{Sil}).      
\end{remark}

\begin{remark}
As we have explained, the duality transformation \eqref{dualmat2} is a particular element of the Weyl group of type $D_4$. At the same time, the same duality matrix also appears in \cite{SW942} as an outer automorphism of $D_4$. The difference is that in \cite{SW942}, the action is written in another basis of the weight space of $D_4$.  Specifically, the mass parameters in \cite{SW942} are components in an orthogonal basis, and the couplings provide mass parameters for the $(A_1)^4$ subalgebra of $D_4$. When the duality action in \cite{SW942} is translated into action on the coupling, it swaps two $A_1$ subalgebras and hence the 8-dimensional vector and spinor irreducible representations, $8_v$ and $8_s$, in agreement with \cite{Gaiotto_Neq2Duality}. 
  
\end{remark}

\section{Quantum Hamiltonians}
\label{a:elliptic}

Here we write explicitly the hamiltonians in elliptic form 
$\hh=\pp^{\,m}+A_2\, \pp^{\,m-2}+\dots+A_m$. For $m=2$ this is the well-known elliptic form of the Heun equation. For $m=3,4$ the formulas are essentially the same as in \cite{EFMV11ecm} (for $m=4$ there is a mistake in \cite{EFMV11ecm}, corrected in \cite{Buric:2021ttm}). For $m=6$ this is a new result. We use the notation $\omega_{1,2,3}$ and $\eta_{1,2}$ for the fixed points, as in \eqref{table0}. 

\paragraph{Case $m=2$:}
In this case, $\tau=\omega_2/\omega_1$ is arbitrary. We denote $g_i:=c_1(\omega_i)$, $i=0\dots 3$; then 
\begin{equation}
    \hh=\pp\,^2-\sum_{i=0}^3 g_i(g_i-\hbar)\wp(q-\omega_i)\,.
\end{equation}

\paragraph{Case $m=3$:}
 
In this case $\omega_2/\omega_1 = \exp(\pi i/3)$, and we have parameters $\mu_j(\eta_i)$, $i, j=0,1,2$, where we set $\eta_0=0$. 
The hamiltonian  has the form
\begin{equation}
    \hh=\pp^{\,3}+(a_2 \wp (q)+b_2 \wp \left(q-\eta_1\right)+ c_2 \wp \left(q-\eta_2\right))\pp+\frac{1}{2} \left(a_3 \wp '(q)+ b_3 \wp '\left(q-\eta_1\right)+ c_3 \wp '\left(q-\eta_2\right)\right)\,.
\end{equation}
Here $a_i, b_i, c_i$ are functions of parameters attached to points $\eta_0, \eta_1, \eta_2$, respectively. Namely, $a_i$ is the $i$th elementary symmetric function of  $\mu_0(\eta_0), \mu_1(\eta_0)+\hbar, \mu_2(\eta_0)+2\hbar$, and similarly for $b_i$ and $c_i$. 

\paragraph{Case $m=4$:}

In this case $\omega_2/\omega_1 = \exp(\pi i/2)$, and we put $\omega_0=0$ for convenience. We have parameters $\mu_j(\omega_i)$, $i=0,3$, and $\mu_j(\omega_{1})=\mu_j(\omega_2)$, with $j=0,1,2,3$, and with $\mu_j(\omega_{1,2})=\mu_{j+2}(\omega_{1,2})$. Recall that $\sum_j\mu_j(x_i)=0$ for each fixed point $x_i$.
The hamiltonian has the form $\hh=\pp^{\,4}+A_2\,\pp^{\,2}+A_3\,\pp+A_4$ with the coefficients as follows (cf. \cite{EFMV11ecm, Buric:2021ttm}): 
\begin{align}
     {A}_2  & = a_2 \wp (q)+b_2 \wp \left(q-\omega _1\right)+2 c_2 \left(\wp \left(q-\omega _2\right)+\wp \left(q-\omega _3\right)\right), \\
     {A}_3  & = \frac{1}{2} a_3 \wp '(q)+\frac{1}{2} b_3 \wp '\left(q-\omega _1\right)+2 \hbar  c_2 \left(\wp '\left(q-\omega _2\right)+\wp '\left(q-\omega _3\right)\right), \\
     {A}_4  & = a_4 \wp(q)^2+b_4 \wp \left(q-\omega _1\right)^2 +c_2 \left(c_2 + 6 \hbar ^2\right) \left(  \wp(q-\omega _2)^2 + \wp(q- \omega_3)^2 \right) \notag \\
     & \qquad + \left( a_2 - b_2\right) \wp(\omega_2) \left(\wp(q- \omega_2) - \wp(q-\omega_3)\right)\,.
\end{align}
The seven parameters $a_{2,3,4}$, $b_{2,3,4}$, $c_2$ are related to $\mu_j(x_i)$ by 
\begin{align}
    a_i&=\sigma_i(\mu_0(0), \mu_1(0)+\hbar,  \mu_2(0)+2\hbar, \mu_3(0)+3\hbar)\,,
    \\
    b_i&=\sigma_i(\mu_0(\omega_3), \mu_1(\omega_3)+\hbar, \mu_2(\omega_3)+2\hbar, \mu_3(\omega_3)+3\hbar)\,,
    \\
c_2&=\sigma_2(\mu_0(\omega_{1,2}),\mu_1(\omega_{1,2})+\hbar)\,,
\end{align}
where $\sigma_i$ denotes the $i$th elementary symmetric function.
\paragraph{Case $m = 6$:}

In this case $\omega_2/\omega_1=e^{\pi i/3}$, and we have six parameters $\mu_j(0)$, $j=0,\dots, 5$, further three parameters 
$\mu_j(\eta_{1})=\mu_j(\eta_2)$, $j=0,1,2$, and two parameters $\mu_j(\omega_{1})=\mu_j(\omega_2)=\mu_j(\omega_{3})$, $j=0,1$. We have $\sum_j\mu_j(x_i)=0$ for each fixed point. We extend $\mu_j(x_i)$ by $\mu_j(\eta_{1,2})=\mu_{j+3}(\eta_{1,2})$ and $\mu_j(\omega_{1,2,3})=\mu_{j+2}(\omega_{1,2,3})$.
The hamiltonian has the form
\begin{equation}
    \hh=\pp^{\,6}+A_2 \, \pp^{\,4}+A_3 \, \pp^{\,3}+A_4 \, \pp^{\,2}+A_5 \, \pp+A_6\,,
\end{equation}
with coefficients as follows:
\begin{align}
    A_{2} & = a_2 \wp(q) + 2b_2 \sum_{i=1,2} \wp(q-\eta_{i})  + 3 c_2 \sum_{i=1,2,3} \wp(q-\omega_{i})\\
    A_{3} & =  \frac{a_3}{2} \wp'(q) + (b_3+3 b_2 \hbar) \sum_{i=1,2} \wp'(q-\eta_{i}) +  6 c_2 \hbar \sum_{i=1,2,3} \wp'(q-\omega_{i})\\
    A_{4} & = a_4 \wp(q)^2 + (b_2^2 + 9 b_3 \hbar + 18 b_2 \hbar^2) \sum_{i=1,2} \wp (q-\eta_{i})^2 + b_2 \sum_{i=1,2} \beta_i \left(\zeta(q-\eta_{i})+\zeta(\eta_i)\right) \\
    & \quad + 3 c_2 \left(c_2+14 \hbar ^2\right)  \sum_{i=1,2,3} \wp(q-\omega_{i})^2 + 2 c_2 \sum_{i=1,2,3} \gamma_i \wp (q-\omega_{i}) \\
    A_{5} & =  \frac{a_5}{2} \wp'(q) \wp(q) + (b_2 b_3+ b_2^2 \hbar + 18 b_3 \hbar ^2 + 12 b_2 \hbar ^3 ) \sum_{i=1,2} \wp'(q-\eta_{i})\wp(q-\eta_{i}) \\
    & \quad  + \sum_{i=1,2} (b_2 (\delta_i -2 \beta_i  \hbar ) - \beta_i  b_3) \wp(q-\eta_{i}) + 6 c_2 \hbar  \left(c_2+8 \hbar ^2\right) \sum_{i=1,2,3} \wp'(q-\omega_{i})\wp(q-\omega_{i})\\
    & \quad + 2 c_2 \hbar \sum_{i=1,2,3} \gamma_i \wp'(q-\omega_{i}) + c_2 \sum_{i=1,2,3} (\rho_i - 3 \xi_i \hbar) \zeta(q-\omega_{i})\\
    A_{6} & = a_6 \wp^3(q)  + \left(b_3 \left( b_3+3 b_2 \hbar +60 \hbar ^3\right) \right)\sum_{i=1,2} \wp^3 (q-\eta_{i}) + \frac{1}{2}\sum_{i=1,2} \left( b_3 (\delta_i -3 \beta_i  \hbar ) \right) \wp'(q-\eta_{i}) \\
    & \quad + c_2 \left(26 c_2 \hbar ^2+c_2^2+120 \hbar ^4\right) \sum_{i=1,2,3} \wp^3 (q-\omega_{i})  + c_2 \left(c_2+6 \hbar ^2\right) \sum_{i=1,2,3} \gamma_i \wp^2 (q-\omega_{i}) \\
    & \quad + c_2 \sum_{i=1,2,3} (\kappa_i -2 \hbar  (\rho_i -2 \xi_i  \hbar ) )\wp (q-\omega_{i})\,.
\end{align}
In these formulas, the parameters $a_{2,3,4,5,6}$, $b_{2,3}$, $c_2$ are the following elementary symmetric functions of the parameters:
\begin{align}
    a_i&=\sigma_i(\mu_0(0), \mu_1(0)+\hbar,  \mu_2(0)+2\hbar, \mu_3(0)+3\hbar, \mu_4(0)+4\hbar, \mu_5(0)+5\hbar)\,,
    \\
    b_i&=\sigma_i(\mu_0(\eta_{1,2}), \mu_1(\eta_{1,2})+\hbar, \mu_2(\eta_{1,2})+2\hbar)\,,
    \\
c_2&=\sigma_2(\mu_0(\omega_{1,2,3}),\mu_1(\omega_{1,2,3})+\hbar)\,.
\end{align}
The other parameters are expressed in terms of $a_i, b_i, c_i$:
\begin{align*}
    \beta _1 & = \left(a_2-2 b_2-27 c_2\right) \wp '\left(\eta _1\right) &  \quad  \beta _2 &= - \beta _1 & \\
    \gamma _1 & =   \left(a_2-8 b_2-3 c_2\right) \wp \left(\omega _1\right) &  \gamma _2   &= \omega ^{-2}\gamma _1 &  \gamma _3   &= \omega ^{2}\gamma _1\\
    \delta_1 & =  \frac{1}{2}  \left(a_3  -2 b_3 - (6 b_2 +108 c_2) \hbar \right) \wp '\left(\eta _1\right) & \quad \delta_2  & = - \delta_1  &  \\
    \xi_1 & = 3  \left(a_2+16 b_2-3 c_2\right) \wp \left(\omega _1\right)^2 &  \xi_2 &= \omega ^{2} \xi_1 &   \xi_3 &= \omega ^{-2} \xi_1\\
    \rho _1& = 3  \left(a_3+16 b_3 -   \left(a_2-32 b_2+9 c_2\right) \hbar \right) \wp \left(\omega _1\right)^2 & \rho _2  &= \omega ^{2} \rho _1 & \rho _3  & = \omega ^{-2} \rho _1\\ 
    \begin{split}   \kappa _1 & =  ( a_2 \left( c_2-4 b_2 \right) +a_4 +28 b_2 c_2  +16 b_2^2 -6 c_2^2 \\ &  \ \ +72 b_3 \hbar + \left(144 b_2-42 c_2\right)  \hbar ^2  ) \wp \left(\omega _1\right)^2
    \end{split} &  \begin{split} \kappa _2  &= \omega ^{2} \kappa _1 \\ & \end{split}  & \begin{split}
        \kappa _3 & = \omega ^{-2} \kappa _1 \\ & 
    \end{split}
\end{align*}
Recall that here $\omega=e^{\pi i/3}$.

\section{Spectral curves in polynomial form}\label{a:fuchs}

To make it easier to compare quantum and classical curves, we convert them into a polynomial form. 

\subsection{Quantum curves}
This is done by multiplying it from the left by $P(x)^m$ where $P(x)=(x-e_1)(x-e_2)(x-e_3)$ for $m=2$ and $P(x)=(x-e_1)(x-e_2)$ for $m=3,4,6$. We then rearrange the expression using 
\begin{equation}
  \hat y:=mP(x)\hbar\frac{d}{dx}\,.  
\end{equation}
Below we present the results, case by case. In all cases we have one accessory parameter, $z$. The coefficients $\alpha_2$, etc., are related to the local exponents at singular points via the formula \eqref{ai}. One can view $\alpha_2, \dots, \alpha_m$ as indeterminate, and instead use \eqref{ai} to determine $\mu_j(0)$ and the local exponents at $x=\infty$ in terms of $\alpha_i$. 

\medskip

\noindent {\bf Case $m=2$}: the quantum curve in polynomial form is
\begin{align}\label{pqc2}
\left(\hat y+\sum_{i=1}^{3}{(g_i-\hbar)}\prod_{j\ne i}^3{(x-e_j)}\right)
     \left(\hat y-\sum_{i=1}^{3}{g_i}\prod_{j\ne i}^3{(x-e_j)}\right)+{(\alpha_2 x-z)}\prod_{i=1}^3{(x-e_i)}\,.  
\end{align}

\medskip

\noindent {\bf Case $m=3$}: the quantum curve in polynomial form is
\begin{equation}\label{pqc3}
    Y_2Y_1Y_0+{\alpha_2}(x-e_1)(x-e_2)Y_0+{(2\alpha_3x-z)}(x-e_1)(x-e_2)\,,
\end{equation}
where
\begin{equation}
    Y_j=\hat y-{(\mu_j(\eta_1)+j\hbar)}{(x-e_2)}-
    {(\mu_j(\eta_2)+j\hbar)}{(x-e_1)}\,.
\end{equation}

\medskip

\noindent {\bf Case $m=4$}: the quantum curve in polynomial form is
\begin{align}
&Y_3Y_2Y_1Y_0+{\alpha_2}{(x-e_1)(x-e_2)}Y_1Y_0\notag\\\label{pqc4}
+&{2\alpha_3}{(x-e_1)(x-e_2)^2}Y_0+
    {(2\alpha_4(3x-2e_1-e_2)-z)}{(x-e_1)(x-e_2)^2}\,,    
\end{align}
where
\begin{equation}
    Y_j=\hat y-(\mu_j(\omega_3)+j\hbar)(x-e_2)-2(\mu_j(\omega_{1,2})+j\hbar){(x-e_1)}\,.
   \end{equation}
    
\medskip

\noindent {\bf Case $m=6$}: the quantum curve in polynomial form is
\begin{align}
        & Y_5 Y_4 Y_3 Y_2 Y_1 Y_0
     +   {\alpha_2}{ \left(x-e_1\right) \left(x-e_2\right)} Y_3 Y_2 Y_1 Y_0\notag \\
     + &  {2\alpha_3}{ \left(x-e_1\right) \left(x-e_2\right)^2} Y_2 Y_1 Y_0 \notag 
     +   {6\alpha_4}{ \left(x-e_1\right)^2 \left(x-e_2\right)^2} Y_1 Y_0 \notag \\\label{pqc6}
     + &  {24\alpha_5}{ \left(x-e_1\right)^2 \left(x-e_2\right)^3} Y_0 
     +   (24\alpha_6(5 x  -3 e_1-2 e_2)-z){ \left(x-e_1\right)^2 \left(x-e_2\right)^3}\,,
\end{align}
where
\begin{align}
    Y_j = \hat y-2({\mu}_j^{(\eta)}+j\hbar){(x-e_2)} -3({\mu}_j^{(\omega)}+j\hbar){(x-e_1)}\,.
\end{align}

\subsection{Classical curves}\label{a:pencils}

In Section \ref{pencils} we described elliptic pencils of special form. Here we verify that our classical spectral curves fit that description. We also give explicit equations of these pencils in projective coordinates. This will be done case by case.

\paragraph{Case $m=2$:}
The classical limit $\hbar=0$ of \eqref{pqc2} can be written as $Q-zP=0$, where
\begin{align}
Q &=\left(y+\sum_{i=1}^{3}{g_i}\prod_{j\ne i}^3{(x-e_j)}\right)
     \left(y-\sum_{i=1}^{3}{g_i}\prod_{j\ne i}^3{(x-e_j)}\right)+{\alpha_2 x}\prod_{i=1}^3{(x-e_i)}\,,\\
     P &=(x-e_1)(x-e_2)(x-e_3)\,. 
\end{align}
Here we think of $x,y$ as $y=\frac12\wp'(q)p$, $x=\wp(q)$, where $p,q$ are canonical coordinates, $\{p,q\}=1$. This induces the Poisson bracket
\begin{equation}
    \{y,x\}=2(x-e_1)(x-e_2)(x-e_3)\,.
\end{equation}
Rewriting $Q, P$ in weighted homogeneous coordinates $(x:y:w)$ on $\mathbb P^2_{1,2,1}$, we get 
\begin{align}
Q &=\left(y+\sum_{i=1}^{3}{g_i}\prod_{j\ne i}^3{(x-e_jw)}\right)
     \left(y-\sum_{i=1}^{3}{g_i}\prod_{j\ne i}^3{(x-e_jw)}\right)+{\alpha_2 x}\prod_{i=1}^3{(x-e_iw)}\,,\\
     P &=w(x-e_1w)(x-e_2w)(x-e_3w)\,. 
\end{align}
The pencil $Q-zP=0$ intersects the line $x-e_iw=0$ at two points $(e_i:\pm g_i\prod_{j\ne i}(e_i-e_j):1)$. To find its intersection with the line $w=0$, we set $x=1, w=0$ and get 
\begin{equation}\label{int2}
 (y+\gg)(y-\gg)+{\alpha_2}=0\,,\qquad \gg=\sum_{i=1}^{3}{g_i}.   
\end{equation}
Recall that $\alpha_2$ is determined by \eqref{aic}:
\begin{equation}
    (p+g_0q^{-1})(p-g_0q^{-1})=(p-\widetilde{g}q^{-1})(p+\widetilde{g}q^{-1})+\alpha_2q^{-2}\,.
\end{equation}
It tells us that \eqref{int2} can be rearranged as $(y+g_0)(y-g_0)=0$, and so the curves of the pencil pass through the points $(1:\pm g_0:0)$. Therefore, this is a pencil of the type described in Sec.~\ref{pencils2}. It is now straightforward to match $Q$ to the expression \eqref{exq2}.

\paragraph{Case $m=3$:}
The classical limit $\hbar=0$ of \eqref{pqc3} can be written as $Q-zP=0$, where
\begin{align}
    Q&=Y_2Y_1Y_0+{\alpha_2}(x-e_1)(x-e_2)Y_0+{2\alpha_3x}(x-e_1)(x-e_2)\,,\\
    P&=(x-e_1)(x-e_2)\,,\qquad  Y_j=y-{\mu_j(\eta_1)}{(x-e_2)}-
    {\mu_j(\eta_2)}{(x-e_1)}\,.
\end{align}
Here 
\begin{equation}
 x=\frac12\wp'(q)\,,\quad y=\wp(q)p\,,\qquad \{y,x\}=3(x-e_1)(x-e_2)\,,
\end{equation}
Writing $Q,P$ in homogeneous coordinates $(x:y:w)$ on $\mathbb P^2$, we get
\begin{align}
    Q&=Y_2Y_1Y_0+{\alpha_2}(x-e_1w)(x-e_2w)Y_0+{2\alpha_3x}(x-e_1w)(x-e_2w)\,,\\
    P&=w(x-e_1w)(x-e_2w)\,,\qquad  Y_j=y-{\mu_j(\eta_1)}{(x-e_2w)}-
    {\mu_j(\eta_2)}{(x-e_1w)}\,.
\end{align}
The cubic $Q=0$ intersects the line $x-e_1w=0$ at the points $\mu_j(\eta_1)(e_1-e_2)$, and $\mu_j(\eta_1)(e_2-e_1)$ for the line $x-e_2w=0$. To find the intersection with the line $w=0$, we set $x=1$, $w=0$ and get
\begin{equation}
    (y-\mmu_2)(y-\mmu_1)(y-\mmu_0)+\alpha_2(y-\mmu_0)+2\alpha_3=0\,. 
\end{equation}
Using the relation \eqref{aic} (and setting $q=-1$), we see that this factorizes as
\begin{equation}
    (y+\mu_2(0))(y+\mu_1(0))(y+\mu_0(0))=0\,. 
\end{equation}
Hence, the pencil $Q-zP=0$ passes through points $(1:-\mu_j(0):0)$.
Therefore, we recognize this as a pencil of cubics from Sec.~\ref{pencils3}. Finally, the polynomial $Q$ can be rearranged as
\begin{align*}
Q & =y^3+Q_2y+Q_3\,,\\
   Q_2 & =  a_2 \left(x-e_1 w\right) \left(x-e_2 w\right)+b_2(e_1-e_2)w \left(x-e_2 w\right)+c_2(e_2-e_1)w \left(x-e_1 w\right), \\
   Q_3 & = 
    a_3 \left(x-e_1 w\right) \left(x-e_2 w\right)^2-b_3 \left(e_1-e_2\right){}^2 w^2 \left(x-e_2 w\right) - c_3 \left(e_2-e_1\right)^2  w^2 \left(x-e_1 w\right)\,.
\end{align*}
The 6 parameters $a_2, b_2, c_2, a_3, b_3, c_3$ are symmetric combinations of linear masses. Indeed, by intersecting this cubic with the three lines, we find that
\begin{equation}\notag
a_i=\sigma_i(\mu_0(0), \mu_1(0), \mu_2(0))\,,\quad
  b_i=\sigma_i(\mu_0(\eta_1), \mu_1(\eta_1), \mu_2(\eta_1))\,,\quad
  c_i=\sigma_i(\mu_0(\eta_2), \mu_1(\eta_2), \mu_2(\eta_2))\,.   
\end{equation}

\paragraph{Case $m=4$:}
The classical limit of \eqref{pqc4} is $Q-zP=0$, with
\begin{align}
Q&=Y_3Y_2Y_1Y_0+{\alpha_2}{(x-e_1)(x-e_2)}Y_1Y_0\notag\\\notag
+&{2\alpha_3}{(x-e_1)(x-e_2)^2}Y_0+
    {2\alpha_4(3x-2e_1-e_2)}{(x-e_1)(x-e_2)^2}\,,\\\notag
    P&=(x-e_1)(x-e_2)^2\,,\qquad Y_j=y-\mu_j(\omega_3)(x-e_2)-2\mu_j(\omega_{1,2}){(x-e_1)}\,.
\end{align}
Here 
\begin{equation}
 x=\wp^2(q)\,,\quad  y=\frac12\wp'(q)p\,,\qquad \{y,x\}=4(x-e_1)(x-e_2)\,.
\end{equation}
We easily confirm that the quartic $Q=0$ intersects $\ell_1: x=e_1$ at points 
\begin{equation}
p_j\,:\  (x,y)=(e_1, \mu_j(\omega_3)(e_1-e_2))\,,\quad j=0,1,2,3\,,   
\end{equation}
while the intersection with $\ell_2: x=e_2$ consists of two points of multiplicity two, 
\begin{equation}
q_j\,:\ (x,y)=(e_2, 2\mu_j(\omega_{1,2})(e_2-e_1))\,,\quad j=0,1\,,    
\end{equation}
due to repetitions among $\mu_j(\omega_{1,2})$. Working in homogeneous coordinates $(x:y:w)$, we also confirm that the intersection of $Q=0$ with $\ell_0: w=0$ consists of $4$ points, 
\begin{equation}
    r_j=(1:-\mu_j(0):0)\,,\quad j=0,1,2,3\,.
\end{equation}
It remains to check that each of the two points $(x_0,y_0)=(e_2, 2\mu_j(\omega_{1,2})(e_2-e_1))$ is an ordinary double point of the quartic $Q=0$. For this, a simple check confirms that each summand in $Q$ belongs to the ideal generated by $(x-x_0)^2$, $(x-x_0)(y-y_0)$, and $(y-y_0)^2$.     

Finally, here is the quartic $Q=0$ in a symmetric homogeneous form:
\begin{align*}
   Q &= y^4+Q_2y^2+Q_3y+Q_4\,,\\
   Q_2 & = a_2 \left(x-e_1 w\right) \left(x-e_2 w\right)+\left(e_1-e_2\right) w \left(b_2 \left(x-e_2 w\right)+2 c_2 \left(x-e_1 w\right)\right), \\
   Q_3 & = \left(x-e_2 w\right)^2 \left(a_3 \left(x-e_1 w\right)-b_3 \left(e_1-e_2\right) w\right), \\
   Q_4 & = \left(e_1-e_2\right)^2 w^2 \left(c_2 \left(a_2-b_2+c_2\right)+b_4\right) \left(x-e_1 w\right) \left(x-e_2 w\right)+a_4 \left(x-e_1 w\right)^2 \left(x-e_2 w\right){}^2 \\
   & +b_4 \left(e_1-e_2\right)^3 w^3 \left(x-e_2 w\right)+c_2^2 \left(e_2-e_1\right)^3 w^3 \left(x-e_1 w\right)\,.
\end{align*}
Checking how it intersects the lines $\ell_{0,1,2}$, we find that the $7$ parameters $a_{i}$, $b_{i}$, $c_{2}$ are symmetric combinations of the linear masses:
\begin{equation}\notag
a_i=\sigma_i(\mu_0(0), \dots, \mu_3(0))\,,\quad
  b_i=\sigma_i(\mu_0(\omega_3), \dots, \mu_3(\omega_3))\,,\quad
  c_2=4\mu_0(\omega_{1,2})\mu_1(\omega_{1,2})\,.   
\end{equation}

\paragraph{Case $m=6$:}
The classical limit of \eqref{pqc6} is $Q-zP=0$, with
\begin{align}
       Q & = Y_5 Y_4 Y_3 Y_2 Y_1 Y_0
     +   {\alpha_2}{ \left(x-e_1\right) \left(x-e_2\right)} Y_3 Y_2 Y_1 Y_0\notag \\
     + &  {2\alpha_3}{ \left(x-e_1\right) \left(x-e_2\right)^2} Y_2 Y_1 Y_0 \notag 
     +   {6\alpha_4}{ \left(x-e_1\right)^2 \left(x-e_2\right)^2} Y_1 Y_0 \notag \\\notag
     + &  {24\alpha_5}{ \left(x-e_1\right)^2 \left(x-e_2\right)^3} Y_0 
     +   24\alpha_6(5 x  -3 e_1-2 e_2){ \left(x-e_1\right)^2 \left(x-e_2\right)^3}\,,\\\notag
     P & = (x-e_1)^2(x-e_2)^3\,,\qquad Y_j = y-2{\mu}_j^{(\eta)}{(x-e_2)} -3{\mu}_j^{(\omega)}{(x-e_1)}\,.
\end{align}
Here 
\begin{equation}
 x=\wp^3(q)\,,\quad y=\frac12\wp(q)\wp'(q)p\,,\qquad \{y,x\}=6(x-e_1)(x-e_2)\,.
\end{equation}
The sextic $Q=0$ intersects $\ell_1: x=e_1$ at three points of multiplicity two, 
\begin{equation}
p_j\,:\ (x,y)=(e_1, 2\mu_j^{(\eta)}(e_1-e_2))\,,\quad j=0, 1, 2\,,   
\end{equation}
while the intersection with $\ell_2: x=e_2$ consists of two points of multiplicity three, 
\begin{equation}
q_j\,:\ (x,y)=(e_2, 3\mu_j^{(\omega)}(e_2-e_1))\,,\quad j=0,1\,,    
\end{equation}
Working in homogeneous coordinates $(x:y:w)$, we also confirm that the intersection of $Q=0$ with $\ell_0: w=0$ consists of $6$ points, 
\begin{equation}
    r_j=(1:-\mu_j(0):0)\,,\quad j=0,\dots,5\,.
\end{equation}
It remains to check that each of $p_j$ is an ordinary double point, and each of $q_j$ is an ordinary triple point. This follows from the formula for $Q$. For example, taking $q_j=(x_0, y_0)$, it is easy to confirm that each summand in $Q$ belongs to the ideal generated by $(x-x_0)^3$, $(x-x_0)^2(y-y_0)$, $(x-x_0)(y-y_0)^2$, and $(y-y_0)^3$.  

Finally, here is the sextic $Q=0$ in a symmetric homogeneous form:
\begin{align*}
   Q &= y^6+Q_2y^4+Q_3y^3+Q_4y^2+Q_5y+Q_6\,,\\
   Q_2 & = a_2 \left(x-e_1 w\right) \left(x-e_2 w\right)+2b_2\left(e_1-e_2\right) w \left(x-e_2 w\right)-3 c_2(e_1-e_2)w\left(x-e_1 w\right), \\
   Q_3 &= a_3 \left(x-e_1 w\right)\left(x-e_2 w\right)^2 -2 b_3 \left(e_1-e_2\right) w\left(x-e_2 w\right)^2, \\
   Q_4 & = a_2 b_2 \left(e_1-e_2\right) w \left(x-e_1 w\right)\left(x-e_2 w\right)^2-2 a_2 c_2 \left(e_1-e_2\right) w \left(x-e_1 w\right)^2\left(x-e_2 w\right) \\
   & +a_4 \left(x-e_1 w\right)^2 \left(x-e_2 w\right)^2+b_2 c_2 \left(e_1-e_2\right) w \left(x-4 e_1 w+3 e_2 w\right) \left(x-e_1 w\right) \left(x-e_2 w\right) \\
   & +b_2^2 \left(e_1-e_2\right)^2 w^2 \left(x-e_2 w\right){}^2 +3 c_2^2 \left(e_1-e_2\right)^2 w^2 \left(x-e_1 w\right)^2, \\
   Q_5 & = \{ a_3 b_2 \left(e_1-e_2\right) w \left(x-e_1 w\right) \left(x-e_2 w\right) -a_2 b_3 \left(e_1-e_2\right) w \left(x-e_1 w\right) \left(x-e_2 w\right) \\
   & - a_3 c_2 \left(e_1-e_2\right) w \left(x-e_1 w\right)^2+
   a_5 \left(x-e_1 w\right)^2 \left(x-e_2 w\right) \\
   & +b_3 c_2 \left(e_1-e_2\right) w \left(x+2 e_1 w-3 e_2 w\right) \left(x-e_1 w\right) 
   -2 b_2 b_3 \left(e_1-e_2\right)^2 w^2 \left(x-e_2 w\right)\}\left(x-e_2 w\right)^2,  \\
   Q_6 & =  \left(e_1-e_2\right)^4 \left(c_2^2 \left(a_2-2 b_2+2 c_2\right)+b_3^2\right)w^4 \left(x-e_1 w\right) \left(x-e_2 w\right) \\
   & -\left(e_1-e_2\right)^3 \left(c_2^2 \left(a_2-2 b_2+c_2\right)+a_3 b_3-2 b_3^2\right) w^3 \left(x-e_1 w\right) \left(x-e_2 w\right)^2 \\
   & +\left(e_1-e_2\right)^2 \left(c_2 \left(a_4-\left(a_2-b_2\right) \left(b_2-c_2\right)\right)-a_3 b_3+b_3^2\right)  w^2\left(x-e_1 w\right)^2 \left(x-e_2 w\right)^2 \\
   & +a_6 \left(x-e_1 w\right)^2 \left(x-e_2 w\right){}^4+b_3^2 \left(e_1-e_2\right)^5 w^5 \left(x-e_2 w\right) \\
   & - c_2^3 \left(e_1-e_2\right)^5  w^5  \left(x-e_1 w\right)\,. 
\end{align*}
The relations between the $8$ parameters $a_i, b_{2,3}, c_2$ and the linear masses are easily determined by considering how $Q=0$ intersects  the lines $\ell_0, \ell_1, \ell_2$. We find that 
\begin{align}\notag
a_i=\sigma_i(\mu_0(0), \dots, \mu_5(0))\,,\quad
  b_i=\sigma_i(2\mu_0^{(\eta)}, 2\mu_1^{(\eta)}, 2\mu_2^{(\eta)})\,,\quad
  c_2=9\mu_0^{(\omega)}\mu_1^{(\omega)}\,.   
\end{align}

\begin{remark} As explained above, each level set $h(p,q)=z$ of the hamiltonian is isomorphic, as an elliptic curve, to the dual curve $h^\vee(p,q)=z$, see \eqref{ms}. This will not be literally true for the curves as presented in this section. The reason is that 
our choice of $z$ as a parameter for elliptic pencils is not canonical, as it can be
shifted by a constant (depending on the mass parameters). The same is true for the formulas for the hamiltonians $\hat{h}$ and $h$ in Sec. \ref{ellipticform}: they are defined up to an additive constant depending on the mass parameters. Due to this ambiguity, the isomorphism of curves in \eqref{ms} becomes an isomorphism between a curve $Q-zP=0$ in the pencil and the dual curve curve $Q^\vee-z^\vee P=0$, where for $z^\vee=z+\mathrm{const}$ for a suitable constant depending on the mass parameters.  
\end{remark}

\section{Algebraic integrability}

According to \cite{EFMV11ecm}, Theorem 7.1, the differential operator $\hh$ is algebraically integrable for certain values of the parameters $\mu_j(x_i)$. This implies the existence of a family of {\it double-Bloch} eigenfunctions 
\begin{equation}
    \hh\psi=z\psi\,,\qquad \psi=\psi(q,z),\quad z\in\C\,, 
\end{equation}
such that 
\begin{equation}\label{m12}
    \psi(q+2\omega_i, z)=M_i \psi(q,z)\,,\qquad i=1,2\,,
\end{equation}
for some $M_{1}, M_{2}\in\C^*$. There is a procedure for calculating such solutions based on a version of Hermite--Bethe ansatz. This is explained below. For convenience, we put $\hbar =1$.

Recall that if $x_i\in\E$ is a fixed point, with stabiliser $\Z_{m_i}\subset \Z_m$, then $\mu_j(x_i)=\mu_{j+ m_i}(x_i)$. Assume, following \cite{EFMV11ecm}, Section 7, that 
\begin{equation}
   \mu_j(x_i)\in m_i\Z\,. 
\end{equation}
This implies that the arithmetic progressions $\{j+\mu_j(x_i)+m_i\Z_{\ge 0}\}$, $j=0,\dots, m_i-1$ do not overlap. 
Write $-n_i$ for the smallest number among $j+\mu_j(x_i)$. Recall that $\mu_j(x_i)$ sum to zero, therefore $n_i\ge 0$. Now consider the set   
\begin{equation}
S_i:=\Z_{\ge 0}\setminus \cup_{j=0}^{m_i-1}\{n_i+j+\mu_j(x_i)+m_i\Z_{\ge 0}\}\,.    
\end{equation}
Our assumptions imply that $S_i$ is a finite set and $|S_i|=n_i$. 

Denote $n:=\sum_{x_i}n_i$, and consider the following function $\phi(q)$ depending on the parameters $t_1, \dots, t_n, \lambda\in\C$:
\begin{equation}\label{ba1}
\phi(q)=e^{\lambda q}\prod_{r=1}^{n}\sigma(q-t_{r})\,.    
\end{equation}
Let us impose $n$ relations on these parameters ($n_i$ relations for each fixed point $x_i$) as follows:
\begin{equation}\label{ba2}
    \left[\frac{d^{s}}{dq^{s}}\left(\phi(q)e^{-n\eta(x_i)q}\right)\right]_{q=x_i}=0\quad\text{for all}\ s\in S_i\,.
\end{equation}
We will refer to \eqref{ba1}--\eqref{ba2} as the Bethe ansatz equations.

\begin{prop}
For generic solutions $t_1, \dots, t_n, \lambda$ of the Bethe ansatz equations, the function 
\begin{equation}\label{psi}
\psi=e^{\lambda q}\,\frac{\prod_{r=1}^{n}\sigma(q-t_{r})}{\prod_{x_i}\sigma(q-x_i)^{n_i}}    
\end{equation}
is an eigenfunction of the hamiltonian $\hh$, 
\begin{equation}\label{eig}
\hh\psi=z\psi,     
\end{equation}
with some $z\in\C$ determined by $t_1,\dots, t_n, \lambda$. The functions $\psi_l=\psi(\omega^lq)$ with $0\le l\le m-1$ span the solution space to the eigenvalue problem \eqref{eig} for generic $z$.
\end{prop}
\begin{proof}
First, by \cite[Theorem 7.1]{EFMV11ecm} and the general results of \cite{BEG97, CEO03} (see Corollary 5.7 and Theorem 5.9 in \cite{Chalykh08}), for generic $z\in\C$ the solution space to \eqref{eig} is spanned by double-Bloch eigenfunctions.

Let now $t_1, \dots, t_n, \lambda$ be a solution to the Bethe ansatz equation, and $\psi$ be the corresponding function \eqref{psi}. Pick one of the fixed points $x_i$, so that $x_i\equiv \omega^lx_i\  (\mathrm{mod}\, \Gamma)$ for some $l$. It can be checked that the function
\begin{equation}
    w(q):=e^{-n\eta(x_i)q}\prod_{x_i}\sigma(q-x_i)^{n_i}
\end{equation}
transforms under the symmetry about $q=x_i$ as follows: 
\begin{equation}
    w(q)\mapsto \omega^{ln_i}w(q)\quad\text{when}\quad q\mapsto (1-\omega^l)x_i+\omega^lq\,.
\end{equation}
This implies that the formal series for $w(q)$ at $q=x_i$ lies in $(q-x_i)^{n_i}\C[[(q-x_i)^{m_i}]]$.
Together with \eqref{ba2}, this means that the formal Laurent series for $\psi$ at $q=x_i$ belongs to the space
\begin{equation}
    U_i:=\bigoplus_{j=0}^{m_i-1} (q-x_i)^{j+\mu_j(x_i)}[[(q-x_i)^{m_i}]]\,.
\end{equation}
On the other hand, our previous analysis (Proposition \ref{semi}) showed that all solutions to \eqref{eig} should belong to $U_i$. 
It follows that double-Bloch eigenfunctions must be of the form \eqref{psi} and the Bethe ansatz equations must hold. Moreover, if $\psi(q)$ is one such eigenfunction and $\lambda$ is generic, then the functions $\psi_l=\psi(\omega^lq)$ will be linearly independent eigenfunctions for the same $z$. 
This proves that the Bethe ansatz method will produce a basis of eigenfunctions. 
\end{proof}

\begin{remark}
    For the Heun equation ($m=2$), Bethe ansatz in this form appeared in \cite{C07ba}, see also \cite{Takemura2003} for a different form. 
\end{remark}

\bibliographystyle{myJHEP}
\bibliography{mybibs}

\end{document}